\documentclass[12pt,english,numbers,super]{elsarticle}
\usepackage[raggedright]{titlesec} 
\usepackage{amsmath}

\usepackage{ragged2e}

\usepackage{rotating}

\makeatletter
\def\ps@pprintTitle{%
 \let\@oddhead\@empty
 \let\@evenhead\@empty
 \let\@oddfoot\@empty
 \let\@evenfoot\@empty
}
\makeatother

\usepackage{times}
\usepackage{amssymb}
\usepackage{subfigure}
\usepackage{subcaption}
\usepackage{adjustbox}
\usepackage[colorlinks=true]{hyperref}

\usepackage{graphicx,epstopdf}
\usepackage{setspace}
\usepackage{accents}
\usepackage{rotating}
\usepackage{xurl}
\usepackage{fancyhdr}
\fancyfootoffset{\dimexpr\textwidth-\oddsidemargin -.1in\relax} 

\newlength{\footnotemargin}
\usepackage{in dentfirst}
\usepackage{lineno}

\usepackage{float}
\usepackage{geometry}

\usepackage{booktabs}
\usepackage{lipsum}
\usepackage{amsthm}
\newtheorem{theorem}{Theorem}[section]

\usepackage{amsmath}
  \usepackage{graphics} 
\usepackage{graphicx}  
  \usepackage{epsfig} 
 
  \usepackage{float}
  \usepackage{epstopdf}
  \usepackage{caption}
\usepackage{amsmath}
\usepackage{wrapfig}
\usepackage{graphicx}
\usepackage{amssymb}

\usepackage{etoolbox}
\patchcmd{\pprintMaketitle}
  {\fi\hrule}
  {\fi\ifvoid\extrainfobox\else\unvbox\extrainfobox\par\vskip10pt\fi\hrule}
  {}{}

\newenvironment{extrainfo}
  {\global\setbox\extrainfobox=\vbox\bgroup\parindent=0pt }
  {\egroup}
\newsavebox\extrainfobox

\usepackage{babel}

\biboptions{numbers,sort&compress}

\makeatletter
\def\@citess#1{\textsuperscript{[#1]}}
\def\@cite#1#2{\textsuperscript{[#1\if@tempswa , #2\fi]}}
\makeatother

\numberwithin{equation}{section}

\newcommand{\beginsupplement}{%
        \setcounter{table}{0}
        \renewcommand{\thetable}{S\arabic{table}}%
        \setcounter{figure}{0}
        \renewcommand{\thefigure}{S\arabic{figure}}%
     }

\numberwithin{equation}{section}
\begin{document}
\begin{spacing}{1.9}

\begin{frontmatter}

\title{Fertilizers Fuel, Insecticides Stabilize: Resolving the Paradox of Enrichment in Agriculture}

\author{Vaibhava Srivastava$^{1,\dagger,*}$\fnref{fn1},~Jason R. Rohr$^{2,\dagger}$\fnref{}, and ~Rana D. Parshad$^3$\fnref{}}

\address{\textit{ 1) Department of Mathematics and Statistics, University of Massachusetts Amherst, \\ Amherst, MA 01003, USA.} \\
\textit{ 2) Department of Biological Sciences, Environmental Change Initiative, Eck Institute of Global Health,\\ University of Notre Dame, Notre Dame, IN 46556, USA. \\
3) Department of Mathematics, Iowa State University, \\ Ames, IA 50011, USA.}}

\fntext[fn1]{Email: Vaibhava Srivastava (vsrivastava@umass.edu), Jason R. Rohr (jrohr2@nd.edu and jasonrohr@gmail.com) and Rana D. Parshad (rparshad@iastate.edu)}

\cortext[cor1]{Corresponding author; \quad  $^\dagger$\text{These authors contributed equally}}

\begin{extrainfo}
\begin{spacing}{1.9}
\begin{abstract}
    The Paradox of Enrichment (PoE) predicts that increasing resources, such as nutrient inputs like fertilizers or food availability, should destabilize ecological systems, such as crop–pest dynamics, leading to population cycles that can increase the risk of crop failure during environmental shocks. Yet, since the Green Revolution, fertilizer use has surged without widespread evidence of yield instability, challenging the PoE's relevance to modern agriculture. Here, we propose and test a novel resolution: that insecticides, frequently co-applied with fertilizers, act as stabilizing agents that counterbalance enrichment-induced instability. Using a modified PoE model with empirically grounded parameters for three major crop–pest systems—soybean–aphid, wheat–aphid, and cabbage–diamondback moth—we find that fertilizer increases yields, but destabilizes dynamics, whereas insecticides restore stability and ensure more predictable harvests. These findings reveal that insecticides may suppress pests but also play a critical role in stabilizing crop yields in nutrient-enriched agroecosystems, with implications for ecosystem management, eutrophication, conservation biology, and pesticide policy.
    

\end{abstract}
\end{spacing}

\end{extrainfo}

\end{frontmatter}

\section*{Main}


\noindent  The Paradox of Enrichment (PoE)\cite{rosenzweig1971paradox} is widely regarded as a seminal and foundational concept in ecology, particularly in the subfields of population dynamics, predator-prey interactions, and ecosystem stability\citep{mccauley1990predator,gilpin1972enriched,abrams1996invulnerable,luckinbill1974effects,gounand2014paradox}. First introduced by Michael Rosenzweig in 1971\cite{rosenzweig1971paradox}, the PoE challenged prevailing assumptions by showing that increasing resources, such as nutrient inputs or food availability, could destabilize rather than stabilize ecological systems. The concept of resource additions shifting consumer--resource dynamics from temporally stable to destable has been reinforced in several follow-up theoretical, experimental, and field studies \citep{luckinbill1973coexistence,bukovinszky2008direct,fussmann2000crossing,myers1981plant,brunsting1985role}. Since its inception, the PoE has catalyzed ecological interests in nonlinear dynamics and feedback loops, facilitated advances in the integration of theoretical modeling with empirical ecology, and laid the groundwork for later work on resilience, tipping points, and regime shifts in ecosystems, and since has often been cited in discussions of ecosystem management, eutrophication, and conservation biology\citep{may2001stability,luckinbill1973coexistence,bukovinszky2008direct,fussmann2000crossing,myers1981plant,brunsting1985role,roy2007stability}.

The societal consequences of the PoE are most obvious when considered through the lens of agroecosystems containing stable interactions between crops and herbivores. Modern agriculture has long relied on fertilizers and insecticides to enhance crop yields and suppress pests \citep{mueller2012closing,erisman2008century,stewart2005contribution,science_infogrpahics}. While their use is foundational to global food production, the broader ecological consequences of these inputs and the theoretical frameworks that predict their effects are not always fully integrated into agricultural models or policy discussions. In agroecosystems, Rosenzweig’s (1971) model\cite{rosenzweig1971paradox} demonstrates that adding fertilizer to crop systems shifts crop-herbivore dynamics from stable to cyclic, with the amplitude of the cycle increasing with the level of enrichment \citep{bukovinszky2008direct}. This destabilization becomes particularly problematic when coupled with inevitable but stochastic events.  For example, if a drought, flood, or disease outbreak happens to coincide with the trough in these crop cycles, then it could lead to zero crop yields (i.e., extinction). Hence, according to Rosenzweig’s validated model \citep{rosenzweig1971paradox}, fertilizers, which are supposed to increase crop yields, can paradoxically and terrifyingly reduce food production and global food security.

Despite the rapid rise in synthetic fertilizer use since the Green Revolution in the mid-$20^{\text{th}}$ century \citep{evenson2003assessing,pingali2012green}, there is little evidence that nutrient enrichment has destabilized food production or threatened food security. On the contrary, fertilizers are credited with contributing 30\% to 50\% of the increase in global crop yields during this period \citep{smil2004enriching}. This observed reality stands in stark contrast to the predictions of the PoE; as a result, the Paradox has long remained a topic of considerable debate and scrutiny \citep{duce2008impacts,murdoch1998plankton}. Some researchers have sought to identify exceptions to the PoE. For example, an empirical study\citep{murdoch1998plankton} found stable population dynamics in enriched plankton systems, whereas broader reviews report limited evidence of instability despite widespread nitrogen enrichment \citep{duce2008impacts,elser2009nutrient}. Other researchers have attempted to identify overlooked mechanisms that might explain why enrichment has not led to widespread instability in natural or agricultural ecosystems \citep{roy2007stability}. Several mechanisms have been proposed, including traits that reduce consumer access to enriched resources, such as low food quality \citep{urabe1996regulation}, unpalatability \citep{genkai1999unpalatable}, switching behavior \citep{van2001alternative}, and inducible defenses \citep{verschoor2004inducible}. Other mechanisms increase nutrient output or loss from consumers, including cannibalism, parasitism, interference, and trophic complexity. Together, these ecological processes might buffer ecosystems against the destabilizing effects predicted by the PoE. However, most of these mechanisms are unlikely to be prevalent in monoculture crop systems, making them improbable explanations for why the PoE has not undermined global food security. Hence, why the PoE has not materialized in agricultural systems remains a fundamental and outstanding question.

One simple explanation that has surprisingly not yet been proposed or tested is that insecticides often applied alongside fertilizers \citep{smith1995crop} may function as de-riching agents that counteract the destabilizing effects of fertilizers. In this view, fertilizers primarily serve to boost mean crop yields, whereas insecticides may play a complementary role by stabilizing yields, reducing variance, and improving predictability for both farmers and consumers. If this hypothesis is correct, it suggests that the long-standing focus on evaluating insecticides solely based on their impact on mean yields may be misdirected. Their true agronomic value might not only lie in increasing yields, but also in enhancing yield stability over time.

To test this hypothesis, we extended the original PoE coupled ordinary differential equation model by incorporating an insecticide term, allowing us to evaluate whether insecticides could theoretically re–stabilize crop–herbivore dynamics. To enhance real-world applicability, we then parameterized the model for three well-studied and economically important crop–pest systems: (1) Soybean (\textit{Glycine max (L.) Merr}.) and soybean aphid (\textit{Aphis glycines}) \citep{ragsdale2007economic,tilmon2011biology}; (2) Winter wheat (\textit{Triticum aestivum L.}) and bird cherry–oat aphid (\textit{Rhopalosiphum padi L.}) 
\citep{riedell1999winter}; (3) Cabbage and diamondback moth (\textit{Plutella xylostella}) \citep{furlong2013diamondback}. These crops are among the most widely cultivated globally and in the United States. In all three systems, the addition of insecticides re–stabilized consumer–resource dynamics that had been destabilized by nutrient enrichment (Fig.~\ref{fig:s1},\ref{fig:long-term_soy},\ref{fig:long-term_wheat},\ref{fig:long-term_cabbage}), offering a plausible explanation for why the PoE has not manifested in agricultural systems since the onset of the Green Revolution.

\section*{Results}

\subsection*{Stability analysis of crop–pest system}

\noindent We recapitulated the classical findings of Rosenzweig\citep{rosenzweig1971paradox} by showing that the crop–pest system exhibits two qualitatively different long–term behaviors depending on ecological and management parameter values. We showed that there is (1) a stable equilibrium between crop and pest when crop yields are low, but (2) when crop yields are high, a Hopf bifurcation\citep{perko2013differential} is crossed and the crop and pest oscillate proportional to the nutrient levels in the system. These results are consistent with predator-prey theory \citep{murray2007mathematical} and are supported by both local and global stability analysis (Supplementary Theorems \textcolor{blue}{S.1.2}, \textcolor{blue}{S.1.3} and \textcolor{blue}{S.1.5}). The interplay of key parameters reveals contrasting effects on system stability. Increases in fertilizer input ($F$), which boost crop yield, consistently reduce system stability and amplify population oscillations, leading to a greater risk of crop damage. In contrast, higher insecticide–induced mortality ($\gamma$ or $\theta$) has a stabilizing effect, suppressing oscillations and reducing pest pressure. All of these analytical results provide a mechanistic understanding of the parameters influencing crop–pest system dynamics and under what conditions the system transitions between extinction, stability, and population oscillation. The detailed mathematical proof of all the theorems is in Supplementary Section \textcolor{blue}{S.3}.
 
We also conducted a global sensitivity analysis of the theoretical crop–pest model to identify the key parameters driving crop yield and pest population dynamics. In both the continuous \eqref{eq:study_holling_type2} and pulsed insecticide application \eqref{eq:study_holling_type2_control} systems, insecticide-induced mortality ($\gamma$ or $\theta$) and fertilizer-induced carrying capacity ($F$) exert the strongest positive influence on crop yield, followed by the intrinsic crop growth rate ($r$) and baseline carrying capacity ($K_0$). Notably, in the pulsed insecticide application \eqref{eq:study_holling_type2_control} system, the pest attack rate ($\alpha$) shows a strong negative correlation with crop yield, highlighting its detrimental effect. For pest dynamics, insecticide mortality rates ($\gamma$ and $\theta$) are the most effective suppressors in both systems, while pest growth rate ($\beta$) and fertilizer–induced carrying capacity ($F$) positively correlate with pest density in the continuous \eqref{eq:study_holling_type2} system. In the pulsed insecticide application \eqref{eq:study_holling_type2_control} system, pest attack rate ($\alpha$) and natural mortality rate ($\delta$) are positively associated with pest populations. These results emphasize the magnitude and direction of parameter effects, with $\gamma$, $\theta$, $F$, $r$, and $K_0$ playing dominant roles in promoting crop yield, and $\beta$ and $\alpha$ driving pest population changes (Supplementary Fig.~\textcolor{blue}{S.1}).

\subsection*{Long–term dynamics of crop–pest system}

\noindent To evaluate whether the results of our theoretical model were relevant to field conditions, we parameterized our model for three crop–pest systems. Our analyses reveal that fertilizer application, while enhancing crop yield, also introduces ecological instability across all three case studies. In soybean–aphid, wheat–aphid, and cabbage–diamondback moth systems, the addition of fertilizer led to population cycles and increased volatility and reduced predictability in crop and pest dynamics (Figs.~\ref{fig:s1},~\ref{fig:long-term_soy}a,~\ref{fig:long-term_wheat}a,~\ref{fig:long-term_cabbage}a). In contrast, the application of insecticides stabilized these dynamics in both the continuous and pulsed insecticide application models when insecticide efficacy was sufficiently high (Figs.~\ref{fig:s1},~\ref{fig:long-term_soy}b and c,~\ref{fig:long-term_wheat}b and c,~\ref{fig:long-term_cabbage}b and c).

For the soybean–aphid case study, the system \eqref{eq:study_holling_type2} stabilizes with the insecticides such as Besiege, Warrior II, and Lannate (Fig.~\ref{fig:long-term_soy}b). The insecticide, Leverage, however, is less effective: it delays stabilization and results in lower overall crop yield (Fig.~\ref{fig:long-term_soy}b). With pulsed insecticide applications \eqref{eq:study_holling_type2_control} and a pest threshold set at $0.1$ ($20\%$ of the maximum aphid population), the system exhibits similar stabilization across all insecticides, though the resulting yield again depends heavily on efficacy (Fig.~\ref{fig:long-term_soy}c). It should be noted that Soybean aphid dynamics in the literature are often modeled using boom-bust type models \cite{kindlmann2010modelling}. However, in the Rosenzweig-MacArthur-type models we use, appropriately scaling the parameters can produce boom-bust dynamics when the system exhibits cyclical behavior. In this case, the peak of the cycle represents the ``boom" phase, while the trough corresponds to the ``bust". Depending on the aphid species, multiple peaks within one growing season are also possible \cite{kindlmann2010modelling}. The same reasoning can be applied for the other aphid case studies herein as well.

In the wheat–aphid case study, the models stabilize with Karate Z and neonicotinoid–based insecticide applications, regardless of whether the insecticides are continuous or pulsed  (Figs.~\ref{fig:long-term_wheat}b and c). For the wheat–aphid case study, crop yield is positively correlated with insecticide efficacy, reinforcing the importance of selecting appropriate treatments for long–term system stability.

For the cabbage–diamondback moth system, stability in the continuous model \eqref{eq:study_holling_type2} is only achieved when insecticide efficacy exceeds $56.3\%$, as observed with Methomyl (Fig.~\ref{fig:long-term_cabbage}b). Less effective treatments fail to suppress pest cycles and instead lead to recurrent destabilization (Fig.~\ref{fig:long-term_cabbage}d). These patterns are reproduced in the pulse application model (Figs.~\ref{fig:long-term_cabbage}c and e), highlighting that efficacy is a key parameter for the long–term system stability and control of oscillations.

We further examined the generalized Holling type III functional response \citep{holling1959some}, where predator saturation effects are more pronounced (Supplementary Section \textcolor{blue}{S.4}). For $k = 1.1$, $1.2$, and $1.3$, the system (Supplementary System \textcolor{blue}{S.1}) transitions from oscillatory to stable dynamics as the insecticide efficacy parameter $\gamma$ increases. This results in stabilized crop yields, supporting the resolution of PoE in crop–pest systems for both Holling type II and III functional responses (Supplementary Figs.~\textcolor{blue}{S.2} and \textcolor{blue}{S.3}~a and b).

\subsection*{Bifurcation analysis of crop–pest system}
\noindent Across both the baseline ODE—continuous \eqref{eq:study_holling_type2} and pulse \eqref{eq:study_holling_type2_control} application models, we observed a consistent pattern: decreasing the feeding efficiency parameter ($\alpha$) or increasing the insecticide–induced mortality rates (either $\gamma$ or $\theta$) helps re–stabilize the destabilization caused by fertilizer application. In contrast, using the low parameter values of $\theta$ or $\gamma$ (representing weak or ineffective control) results in persistent population oscillations and thus greater variability (Fig.~\ref{fig:soy_biffur} and Supplementary Figs.~\textcolor{blue}{S.4} and \textcolor{blue}{S.5}). These qualitative trends are supported by rigorous mathematical stability analysis, which provides us with theoretical parametric restrictions (Supplementary Section \textcolor{blue}{S.3}). A comparable stabilization effect was also observed under the Holling type III functional response when $k = 1.2$ (Supplementary Fig.~\textcolor{blue}{S.3}~c and d).

Our one–dimensional bifurcation reinforces the finding that the insecticide–induced mortality stabilizes the system, and thus resolves the PoE (Fig.~\ref{fig:soy_biffur} and Supplementary Figs.~\textcolor{blue}{S.4} and \textcolor{blue}{S.5}). We also studied the two–dimensional bifurcation dynamics by varying insecticide efficacy ($\gamma$ or $\theta$) and feeding efficiency ($\alpha$). For the continuous application system \eqref{eq:study_holling_type2}, we identified a critical threshold value for $\gamma$ 
of approximately $0.65, 0.81,$ and $1.3$ for soybean, wheat, and cabbage systems, respectively (Fig.~\ref{fig:soy_biffur}e and Supplementary Fig.~\textcolor{blue}{S.6}~a,c). Once this critical threshold value of $\gamma$ is crossed, the system stabilizes regardless of the values of $\alpha$.

However, for pulsed insecticide applications \eqref{eq:study_holling_type2_control},  at lower values of $\alpha$, even a modest pulse ($\theta$) is sufficient to stabilize the system. However, as $\alpha$ increases (which indicates more efficient feeding behaviour), the system becomes increasingly resistant to stabilization, even with stronger pulsed insecticide–control inputs (Fig.~\ref{fig:soy_biffur}f, and Supplementary Figs.~\textcolor{blue}{S.6}~b and d).

For Holling type III dynamics (Supplementary System \textcolor{blue}{S.1}), we observed the existence of a critical stabilizing threshold. When $\gamma>0.8$, the system stabilizes regardless of feeding efficiency (Supplementary Fig.~\textcolor{blue}{S.3}~e). The pulse application version of this model also shows similar nonlinear interplay between $\theta$ and $\alpha$, reproducing the trends seen in the Holling type II case (Supplementary Fig.~\textcolor{blue}{S.3}~f).

\subsection*{Timing of insecticide spraying}

\noindent The timing of insecticide application can significantly influence the extent of pest–induced damage to the crop and the number of pests that need to be controlled, ultimately affecting both the stability of yield outcomes and the effectiveness of pest management strategies. This relationship is reflected in real–world agricultural practices, such as soybean farming, where insecticide use is typically guided by economic thresholds of pest density \citep{ragsdale2007economic, higley1996economic}. To explore this dynamic, we conducted numerical simulations using a parameterized soybean–aphid crop–pest model to study the effects of insecticide application timing and pest population thresholds that trigger insecticide spraying ($y_{\text{thres}}$) on crop yield stability. We tested pest population threshold values ($y_{\text{thres}}$) of 0.1, 0.2, 0.3, 0.4, and 0.5, corresponding to progressively more tolerant intervention strategies. These thresholds represent increasing tolerance levels for aphid density before triggering an insecticide spraying event. To examine the role of insecticide persistence, we explicitly marked the timing of spray applications across three insecticide efficacy regimes (Supplementary Figs.~\textcolor{blue}{S.7}, \textcolor{blue}{S.8}, \textcolor{blue}{S.9}), thereby highlighting their temporal distribution and impact on pest suppression dynamics.

When the pest threshold is set to $0.1$ ($20$\% of the maximum aphid population) or $0.2$ ($40$\%), the system achieves a stable crop yield (Supplementary Figs.~\textcolor{blue}{S.7}~a,~b) when assuming a $90\%$ mortality rate for aphids only on the day of spraying. Although a threshold of $0.3$ results in a comparatively lower yield than $y_{\text{thres}} = [0.1, 0.2]$, but the system still exhibits stabilization (Supplementary Fig.~\textcolor{blue}{S.7}~c). However, higher thresholds of $0.4$ and $0.5$ fail to stabilize the system, leading to ineffective pest suppression and destabilization of the soybean–aphid system (Supplementary Figs.~\textcolor{blue}{S.7}~d, e).

A similar trend is observed when the insecticide maintains a residual efficacy of $70$\% on the day following application. In this case, pest thresholds of $0.1$, $0.2$, and $0.3$ still support yield stabilization, whereas a threshold of $0.4$ results in partial stabilization but at a significantly reduced yield (Supplementary Figs.~\textcolor{blue}{S.8}~a--d). At a threshold of $0.5$, the system remains destabilized and yields remain critically low (Supplementary Figs.~\textcolor{blue}{S.8}~e).

The most effective outcome, characterized by both stable crop yield and low pest density, is achieved when insecticide efficacy decays gradually across three days: $90\%$ on the first day, $80\%$ on the second, and $50\%$ on the third (Supplementary Figs.~\textcolor{blue}{S.9}~a--e). Insecticides that stay effective for a longer time help keep pest populations low for longer periods. This leads to more stable crop yields because there is less chance for pests to bounce back quickly after spraying. Our results show that long–lasting insecticides are better at keeping the system stable than those that are strong at first but wear off quickly. So, the duration of insecticide effectiveness matters more for stability than just the timing of when it is applied.

\section*{Discussion}

\noindent In this study, we investigated a longstanding ecological puzzle—the Paradox of Enrichment—through the lens of modern agriculture. While classical theory predicts that nutrient enrichment should destabilize consumer–resource dynamics \citep{rosenzweig1971paradox}, real–world crop systems have defied this expectation. Despite a dramatic rise in fertilizer use since the Green Revolution, global food production has remained relatively stable and robust, even in the face of environmental variability \citep{evenson2003assessing,pingali2012green}. Here, we propose and test a novel explanation: that insecticides, commonly applied alongside fertilizers, serve as stabilizing agents that counteract the destabilizing effects of nutrient enrichment. 

We extended the classical Rosenzweig–MacArthur predator–prey model \citep{rosenzweig1963graphical} by incorporating insecticide–induced pest mortality. Our analysis included ordinary differential equations \citep{murray2007mathematical,perko2013differential} with both continuous (system \ref{eq:study_holling_type2}) and threshold–triggered (pulse) insecticide applications (system \ref{eq:study_holling_type2_control}) and grounded parameterizations for three major crop–pest systems: (1) soybean and soybean aphid, (2) winter wheat and bird cherry–oat aphid, and (3) cabbage and diamondback moth. In each system, fertilizer increased crop yield but introduced instability in the form of population cycles consistent with the PoE (Figs.~\ref{fig:s1},~\ref{fig:long-term_soy}a, ~\ref{fig:long-term_wheat}a, and \ref{fig:long-term_cabbage}a). However, the addition of insecticides re–stabilized dynamics, provided their efficacy surpassed a system–specific threshold (Figs.~\ref{fig:s1},~\ref{fig:long-term_soy}b and c, ~\ref{fig:long-term_wheat}b and c, and \ref{fig:long-term_cabbage}b and c).

Our mathematical analysis (Figs.~\ref{fig:soy_biffur} and Supplementary Figs.~\textcolor{blue}{S.4}--\textcolor{blue}{S.5}) demonstrates that stability in enriched agroecosystems is critically dependent on insecticide efficacy ($\gamma$ or $\theta$). Below the efficacy threshold, enrichment leads to oscillations in crop and pest densities, potentially compromising yield reliability. Above the threshold, insecticide applications restore stability, even in systems with high pest growth rates. Importantly, these findings held across both Holling type II and III functional responses \citep{holling1959some}, suggesting that the results are robust to different assumptions about pest foraging behavior.

Sensitivity (Supplementary Fig.~\textcolor{blue}{S.1}) and bifurcation analyses (1 dimensional: Supplementary Figs.~\textcolor{blue}{S.4} and \textcolor{blue}{S.5}; 2 dimensional: Supplementary Fig.~\textcolor{blue}{S.6}) revealed that insecticide efficacy interacts nonlinearly with other parameters, particularly pest feeding efficiency ($\alpha$), to determine long-term dynamics. In systems using threshold–triggered insecticide applications, stabilization was more difficult to achieve when $\alpha$ was high, underscoring the importance of both ecological interactions and management practices. We also showed that the persistence of insecticide effectiveness across multiple days (rather than a single pulse) improved yield, highlighting the role of formulation and environmental degradation in pest control efficacy (Supplementary Figs.~\textcolor{blue}{S.7}--\textcolor{blue}{S.9}).

Crucially, our work reframes how insecticides are evaluated in agroecosystems \citep{dent2020insect}. Traditionally, insecticides have been assessed based on their ability to increase average crop yields by minimizing yield losses caused by pests. \citep{oerke2006crop,savary2019global}. However, our models suggest that a principal benefit may lie in stabilizing yields—reducing temporal variability and ensuring predictable harvests in the face of nutrient enrichment and pest outbreaks. This reframing has implications not only for agronomic decision-making but also for how insecticide regulations and policies are developed in the context of food security and environmental sustainability \citep{tilman2002agricultural,mohring2020pathways}.

From a broader ecological perspective, our results help resolve an important theoretical inconsistency between classical ecological models and real–world agricultural dynamics. While previous work has proposed mechanisms such as low food quality \citep{urabe1996regulation} and induced defenses \citep{verschoor2004inducible} to buffer natural systems against enrichment-induced instability, these mechanisms are often absent or minimized in simplified, monoculture–dominated agroecosystems. Our findings suggest that anthropogenic inputs namely, insecticides serve a critical, underappreciated, stabilizing role in these systems. Despite being parameterized, our modeling approach remains theoretical, underscoring the need for empirical validation. To test the hypothesis that insecticides stabilize invertebrate consumer–resource dynamics, we encourage experiments that cross nutrient and insecticide additions and track both consumer and resource abundance and variance through time. Alternatively, comparing pest pressures and annual crop yields between organically grown systems (receiving fertilizer only) and conventionally managed systems (fertilizer plus synthetic insecticides) could provide valuable evidence for the yield-stabilizing effects of insecticides in agriculture.

At the same time, the environmental costs of insecticide use are well documented, including impacts on non–target species, soil microbiota, and water quality \citep{banerjee2023soil,kohler2013wildlife,stehle2015agricultural}. Our model highlights that only moderate efficacy levels are needed to restore stability, suggesting that lower doses or targeted applications may be sufficient in many cases. This opens up opportunities for optimizing pest management strategies to balance yield stability with ecological sustainability \citep{wyckhuys2020ecological,barzman2015eight,pretty2018intensification}.

Finally, while our model is intentionally simple, it is designed to be flexible and extensible. Future work could incorporate more complex biological interactions such as predator–prey–patch dynamics \citep{simon2011toxicology}, biocontrol agents \citep{bahlai2015shifts}, pest resistance \citep{bras2022pesticide}, and economic and optimal costs of control strategies \citep{ragsdale2007economic,mulungu2024economic}. These extensions could further inform the design of resilient agroecosystems that are both productive and environmentally responsible.

\section*{Methods}

\subsection*{Model}


\noindent To model the dynamics between crops and pests, we used classical predator-prey dynamics based on the well-known Rosenzweig–McArthur predator–prey framework\citep{rosenzweig1963graphical}. Specifically, we considered a Holling type II response\citep{holling1959some}, which introduces a saturating effect of predation in the model. For our work, we developed an analogous crop–pest model where the crop yield can be increased by adding fertilizer, and insecticide usage induces mortality of the insect pest.

\subsubsection*{Crop–pest system}
\noindent Let $(x)$ and $(y)$ denote the crop and pest in the agroecosystems. In our system, a crop $(x)$ follows a logistic growth with a growth rate $(r)$, and it can reach to carrying capacity $(K_0)$ without fertilizer. Using fertilizer results in an increased yield of $(F)$. The pest $(y)$ attacks the crop with rate $(\alpha)$, which has a growth rate of $(\beta)$, and with a natural mortality rate of $(\delta)$. Use of insecticide results in insecticide-induced mortality ($\gamma$). The $\gamma$ is calculated as $-\ln\text{(survival})/\text{time}$.  The details of the parameters and their units are in Table \ref{tab:1}. The crop–pest dynamics is modeled using the system of nonlinear ordinary differential equations:

\begin{equation}\tag{1}
	\begin{array}{l}
		\dfrac{dx}{dt} = rx \left( 1 - \dfrac{x}{K_0 + F} \right) - \dfrac{\alpha x y}{1 + hx}, \\[12pt]
\dfrac{dy}{dt} = \dfrac{\beta x y}{1 + hx} - \delta y - \gamma y.
	\end{array}
	\label{eq:study_holling_type2}
\end{equation}

One limitation of the system in \eqref{eq:study_holling_type2} is that it assumes insecticide spraying occurs continuously. To increase the realism, we provide a pulse version that relies on a pest threshold, meaning that insecticide application is only necessary when the pest population exceeds this threshold \citep{ragsdale2007economic}. The pulse formulation of the model is triggered by the pest population threshold $(y_{\text{thresh}})$

\begin{equation}\tag{2}
	\begin{array}{ll}
		\dfrac{dx}{dt} = rx \left( 1 - \dfrac{x}{K_0 + F} \right) - \dfrac{\alpha x y}{1 + hx}, \\[12pt]
		\dfrac{dy}{dt} = \dfrac{\beta x y}{1 + hx} - \delta y , \\[12pt]
		\text{If } y > y_{\text{thresh}}, \hspace{.05in} \text{then } y \rightarrow (1-\theta) y,
	\end{array}.
	\label{eq:study_holling_type2_control}
\end{equation}
where $\theta$ is the mortality (in percentage) caused by insecticide to the pest. The system \eqref{eq:study_holling_type2_control}, helps us in setting a pest threshold, which triggers the insecticide spraying usage.

\subsection*{Techniques}

\noindent\textbf{Stability analysis:} We analyzed the stability of the crop–pest system using techniques from dynamical systems theory, including linear stability analysis, global stability analysis, and other standard methods for ordinary differential equations \citep{murray2007mathematical,perko2013differential}. We derived parametric conditions that govern the system’s qualitative behavior and gained mechanistic insight into how parameters influence transitions between extinction, stable coexistence, and population oscillations. The stability analyses provide a deeper understanding of the thresholds that determine stability regimes in the crop–pest interaction. Full mathematical derivations and proofs of all theorems are in Supplementary Section \textcolor{blue}{S.3}.

\noindent\textbf{Global sensitivity analysis:} To understand the influence of model parameters on the crop–pest system for both continuous \eqref{eq:study_holling_type2} and pulse \eqref{eq:study_holling_type2_control} insecticide application systems, we performed a correlation–based global sensitivity analysis using both Pearson and Spearman correlation coefficients \citep{bewick2003statistics}. The correlation-based analysis helps us identify which parameters most significantly influence the crop and pest densities. Pearson correlation captures linear associations, while Spearman correlation detects monotonic (including nonlinear) relationships, enabling a comprehensive assessment of parameter sensitivity. 

We first generated parameter sets using Latin Hypercube Sampling (LHS) \citep{loh1996latin}, which ensures efficient and representative sampling across the multidimensional parameter space. For each sampled set, both systems \eqref{eq:study_holling_type2} and \eqref{eq:study_holling_type2_control} were simulated, and the resulting crop and pest densities were recorded. Correlation coefficients were then computed between the sampled parameters and model outputs to identify parameters with positive and negative correlations in the systems.

\noindent\textbf{Bifurcation Analysis: } To understand the mechanisms underlying crop–pest system stabilization, we conducted one– and two–dimensional bifurcation analyses. Bifurcation analysis is a mathematical technique used to study how the qualitative behavior of a dynamical system changes as model parameters are varied, particularly focusing on transitions in stability \citep{murray2007mathematical,perko2013differential}.

For one–dimensional bifurcation, we focus on key control parameters, insecticide efficacy $\gamma$ for system \eqref{eq:study_holling_type2} and $\theta$ for system \eqref{eq:study_holling_type2_control}. Given the limited empirical data on the pest's feeding efficiency ($\alpha$), we also explored how interactions between insecticide efficacy and feeding traits shape the dynamics. To understand the interplay of $\alpha$ and insecticide efficacy parameters (either $\gamma$ or $\theta$), we used the two–dimensional bifurcation analysis.

\subsection*{Background on focal crop-pest systems}
\noindent We focused on three economically significant crop–pest systems: soybean (\textit{Glycine max (L.) Merr.}) with soybean aphid (\textit{Aphis glycines}), winter wheat (\textit{Triticum aestivium L.}) with bird cherry-oat aphid (\textit{Rhopalosiphum padi L.}), and cabbage with diamondback moth (\textit{Plutella xylostella}). These systems were selected due to their broad geographic cultivation, substantial economic importance, and vulnerability to key insect pests that vary in feeding guild and epidemiological role. Soybean and wheat are globally dominant field crops\citep{usdaSoybeans2025,usdaWheat2025}, whereas cabbage represents a major vegetable crop with high–value production\citep{FAO2001cabbage}. The chosen pests include both direct-damaging herbivores and virus vectors, collectively responsible for billions of dollars in crop losses annually \citep{ragsdale2007economic,tilmon2011biology,perkins2018impact,dunn2024regional}. A detailed background on each focal crop–pest system is provided in Supplementary Section \textcolor{blue}{S.1}.

\subsection*{Parametrization of case studies}
\subsubsection*{Soybean (\textit{Glycine max (L.) Merr.}) and Soybean Aphid (\textit{Aphis glycines})}
 
\noindent Soybean crop growth rate and carrying capacity are parameterized based on yield data\citep{USDA_Soybean_Production,ball2000optimizing,cordova2020soybean}. The growth rate is estimated at $0.451 \text{day}^{-1}$\citep{ball2000optimizing}. The crop yield without fertilizer is $K_0 = 2.797 ~\text{T/Ha}$\citep{cordova2020soybean}, with yield increases under different nitrogen fertilization strategies ($F=0.559$ T/Ha and $0.291$ T/Ha)\citep{cordova2020soybean}. Soybean aphid population growth is modeled exponentially with a rate $\beta\approx \frac{\ln 2}{2}=0.3466 ~\text{day}^{-1}$, consistent with population doubling every two days under favorable conditions\citep{ragsdale2007economic,tilmon2011biology}. The natural mortality rate of aphid is $\delta = 0.05 ~\text{day}^{-1}$\citep{ragsdale2007economic, tilmon2011biology, pioneerSoybeanAphidMgmt}. The parameterization of the insecticide mortality rate is based on a study conducted at the South Dakota State University Volga Research Farm in 2017 \citep{dierks2019evaluating, menger2020implementation} (See Supplementary Table \textcolor{blue}{S.1} ).

\noindent \subsubsection*{Winter wheat (\textit{Triticum aestivium L.}) and bird cherry–oat aphid (\textit{Rhopalosiphum padi L.})}

\noindent For the wheat crop, we model growth using a logistic curve with maximum yield $K = 4.185$ T/Ha, determined from average Tennessee yields over $2006$-$2016$ \citep{perkins2018impact}. The baseline unfertilized yield is set to $K_0 = 1.877$ T/Ha, and the fertilization gain is $F = 2.308$ T/Ha, based on field trials in Idaho and Washington that showed a $123\%$ increase in yield under ammonium nitrate treatments \citep{mahler1994nitrogen}. Growth rate is estimated from logistic model fitting using the time to reach 99\% of yield ($\sim$115 days), giving $r = 0.0722$ day$^{-1}$ (Supplementary Section \textcolor{blue}{S.1.2}). The bird cherry–oat aphid has growth and mortality rates similar to those of the soybean aphid, using $\beta = 0.3466$ day$^{-1}$ and $\delta = 0.05$ day$^{-1}$ \citep{ragsdale2007economic, tilmon2011biology, pioneerSoybeanAphidMgmt}. Insecticide efficacy data is taken from late-winter foliar spray trials \citep{perkins2018impact}.

\noindent \subsubsection*{Cabbage and Diamondback moth (\textit{Plutella xylostella})}

\noindent To parameterize the effect of fertilization on cabbage, we have used a field experiment, where researchers examined the impact of various fertilization and irrigation methods on the yield of white cabbage \citep{vsturm2010effect}, conducted in the Ljubljansko polje region, east of Ljubljana, the capital of Slovenia. The study found that the yield in unfertilized control plots, using the farmers' irrigation practices, was $K_0 = 47$ T/Ha. When farmers applied fertilization practices in accordance with the \textit{Regulations on Integrated Production of Vegetables} and the \textit{Technological Instructions for Vegetable Production} \citep{OGRS2002}, in addition to irrigation, the yield increased by $97.87\%$, reaching a total of $93$ T/Ha (i.e., $F = 46$ T/Ha). Using a logistic–type growth model, we estimate the intrinsic growth rate of cabbage as $r = 0.091$ day$^{-1}$. The growth rate of diamondback moth is $\beta=0.3466  \hspace{.051in} \text{day}^{-1}$ \citep{kahuthia2008development} and  $\delta=0.05 \hspace{.1in}\text{day}^{-1}$  \citep{al2019role}. 

To understand the effect and efficacy of multiple registered active ingredients in insecticides on diamondback moth, we used the findings of a field study conducted in Georgia and Florida (See Supplementary Table \textcolor{blue}{S.2})\citep{dunn2024regional}. We also evaluated the effectiveness of the baculovirus Plutella xylostella Granulovirus (PlxyGV) for managing this pest based on research conducted in Kabete and Thika, near Nairobi \citep{bhattacharyya2006pest}.

\subsection*{Timing of insecticide spraying}

\noindent Insecticides are typically recommended only when pest populations exceed a defined economic threshold \citep{ragsdale2007economic}. This method has proven successful in lowering production costs, minimizing environmental hazards from insecticides, and aims to increase profits \citep{mumford1984economics}. This approach aligns with real-world practices, such as those observed in soybean farming, where farmers typically apply insecticides only after aphid counts surpass critical levels. For soybean crops, once aphids are detected, their population per plant should be estimated. If the count reaches around $100$ to $250$ aphids per plant, advisable to apply insecticide \citep{higley1996economic,ragsdale2007economic}.

We used a parameterized case study of soybean aphid management, where insecticide application is triggered once the aphid population exceeds a specified threshold, $y_{\text{thres}}=[0.1,0.2,0.3,0.4,0.5].$ We simulate three insecticide effectiveness scenarios: (1) The insecticide mortality is $90\%$ and is only effective on the day of spraying, (2) The insecticide mortality is $90\%$ on the day of spraying and $70\%$ for the following day, and (3) The insecticide mortality is $90\%, 80\%$, and $50\%$ for three days. 

\subsection*{Data availability}
\noindent No original data were used in this manuscript, and all search terms and articles from the literature review are described and cited.


\subsection*{Acknowledgments}
\noindent This research was supported by grants from the National Science Foundation (DEB-2017785, DEB-2109293, BCS- 2307944, ITE- 2333795), US Department of Agriculture (2021-38420-34065), the Frontier Research Foundation, and the University of Notre Dame Poverty Initiative to JRR. RDP acknowledges support by the Agricultural and Food Research Initiative Grant no. 2023-67013-39157 from the USDA National Institute of Food and Agriculture.

\subsection*{Author contributions}

\noindent V.S.: conceptualization, formal analysis with assistance from J.R.R. and R.D.P., methodology, software, visualization, writing—original draft, writing—review and editing; J.R.R.: conceptualization, methodology, supervision, writing—original draft, writing—review and editing; R.D.P.: conceptualization, methodology, supervision, visualization, writing—original draft, writing—review and editing. All authors gave their final approval for publication.

\subsection*{Competing interests}
\noindent The authors declare no conflicts of interest.

\subsection*{Additional information}

\noindent\textbf{Supplementary information:} The supplementary material for review is attached after the reference section.

\noindent\textbf{Correspondence and requests for materials} should be addressed to Vaibhava Srivastava or Jason R. Rohr.

\bibliography{ref.bib}

\newpage

\begin{table}[H]
    \caption{\textbf{Model parameters and values used in crop–pest systems.}}
    \centering
    \resizebox{\textwidth}{!}{%
    \begin{tabular}{|l|l|c|c|c|l|}
        \hline
        \textbf{Parameter} & \textbf{Description} & \textbf{Cabbage-DBM} & \textbf{Wheat-Aphid} & \textbf{Soybean-Aphid} & \textbf{Units} \\
        \hline
        $r$        & Crop growth rate                         & 0.0911\citep{murray2007mathematical} & 0.0722\citep{murray2007mathematical} & 0.451\citep{ball2000optimizing} & day$^{-1}$ \\
        $K_0$      & Base carrying capacity (no fertilizer)   & 47\citep{vsturm2010effect}     & 1.877\citep{perkins2018impact,mahler1994nitrogen}   & 2.797\citep{cordova2020soybean} & tons/ha \\
        $F$   & Increased Carrying capacity (with fertilizer)               & 46\citep{OGRS2002} & 2.3080\citep{perkins2018impact,mahler1994nitrogen} & \begin{tabular}[c]{@{}l@{}} 0.291\cite{cordova2020soybean}\\0.559\cite{cordova2020soybean}\end{tabular} & tons/ha \\
        $\alpha$   & Maximum competition rate (pest on crop)     & 0.3    & 1      & 2.5   & unitless \\
        $\beta$    & Pest growth efficiency                   & 0.1477\citep{kahuthia2008development} & 0.3466\citep{ragsdale2007economic,tilmon2011biology} & 0.3466\citep{ragsdale2007economic,tilmon2011biology} & day$^{-1}$ \\
        $\delta$   & Pest natural mortality rate              & 0.1\citep{al2019role}    & 0.05\citep{pioneerSoybeanAphidMgmt,tilmon2011biology}   & 0.05\citep{pioneerSoybeanAphidMgmt,tilmon2011biology}  & day$^{-1}$ \\
        $h$        & Handling time                            & 0.1    & 0.1    & 0.1   & days \\
        \hline
    \end{tabular}%
    }
    \caption*{\footnotesize \\
Detailed description of the parameters used in the crop–pest system (continuous version: \eqref{eq:study_holling_type2} and pulse application version: \eqref{eq:study_holling_type2_control}), including their ecological description and the parameterized values applied for three case studies: (1) Cabbage–Diamondback Moth (DBM), (2) Wheat–Aphid, and (3) Soybean–Aphid systems.}
    \label{tab:1}
\end{table}

\newpage

\begin{figure}[H]
\includegraphics[width=1\linewidth]{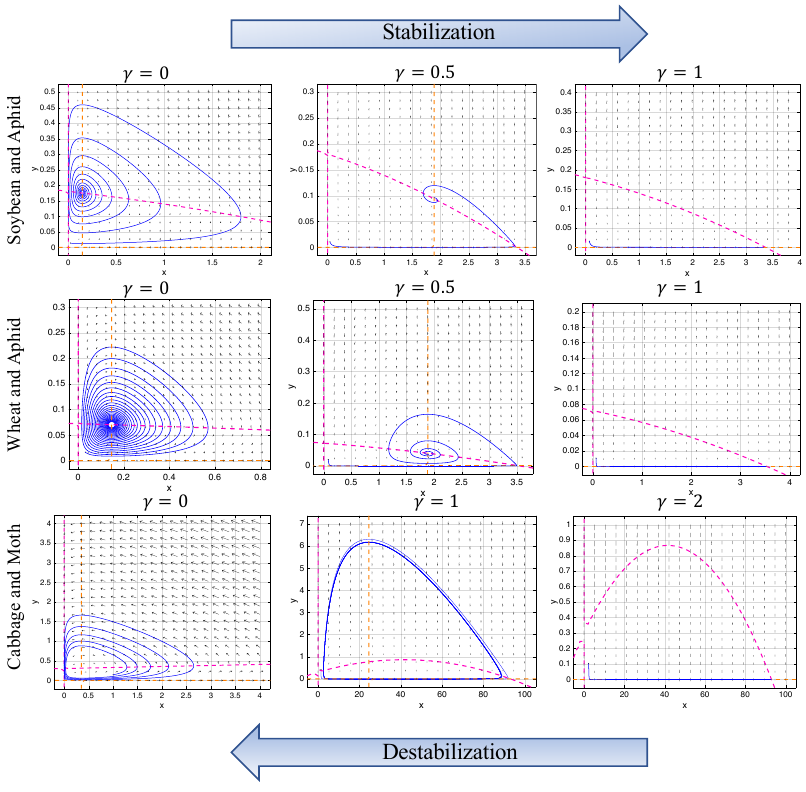}
   \caption{\textbf{Phase diagram illustrating the impact of insecticide-induced mortality parameter $(\gamma)$ on the stabilization and destabilization of crop–pest systems.} Phase diagram showing the insecticide-induced mortality parameter $(\gamma)$, which stabilizes the crop-pest dynamics for all three major crop–pest systems: (1) Soybean and Soybean Aphid, (2) Winter wheat and bird cherry-oat aphid, and (3) Cabbage and Diamondback moth. In the phase diagram, the magenta curve represents the crop-nullcline, the orange curve represents the pest-nullcline, the intersection of nullclines represents equilibrium points, and the blue curve illustrates the population trajectory. The system parameters are listed in Table \eqref{tab:1}. (Note: A phase plot is a geometric visualization of the trajectories of an ordinary differential equation (ODE), showing how the system evolves over time from a given initial condition. A nullcline is a curve in the phase plane where the rate of change of one of the variables in the ODE is zero.) }
    \label{fig:s1}
\end{figure}

\newpage

\begin{figure}[H]
\includegraphics[width=1\linewidth]{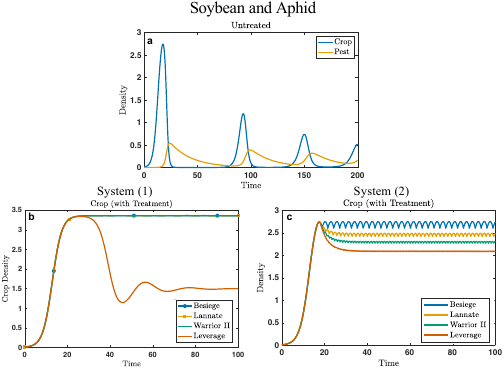}
   \caption{\textbf{The role of insecticide in stabilizing crop yield in the soybean-aphid system.} The long-term dynamical behavior of the soybean-aphid crop-pest model. \textbf{a}, Application of fertilizer caused destabilization in the system. \textbf{b} and \textbf{c}, Application of insecticides stabilized the soybean-aphid crop-pest model. The system parameters are listed in Table \eqref{tab:1}.}
    \label{fig:long-term_soy}
\end{figure}

\newpage

\begin{figure}[H]
\includegraphics[width=1\linewidth]{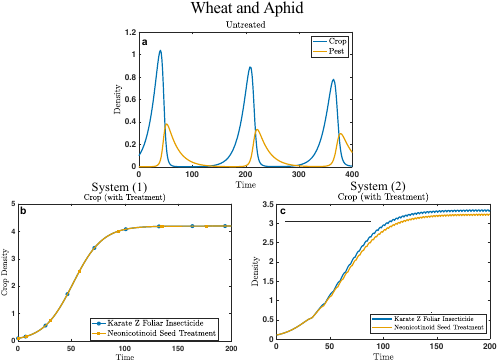}
   \caption{\textbf{The role of insecticide in stabilizing crop yield in the wheat-aphid system.} The long-term dynamical behavior of the wheat-aphid crop-pest model. \textbf{a}, Application of fertilizer caused destabilization in the system. \textbf{b} and \textbf{c}, Application of insecticides stabilized the wheat-aphid crop-pest model. The system parameters are listed in Table \eqref{tab:1}.}
    \label{fig:long-term_wheat}
\end{figure}

\newpage

\begin{figure}[H]
\includegraphics[width=1\linewidth]{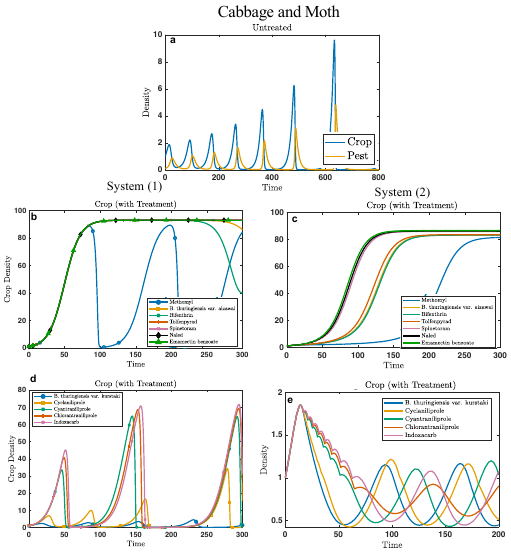}
   \caption{\textbf{The role of insecticide in stabilizing crop yield in the cabbage–moth system.} The long-term dynamical behavior of the cabbage-moth crop-pest model. \textbf{a}, Application of fertilizer led to destabilization in the system. \textbf{b} and \textbf{c}, Use of high-efficacy insecticides helped stabilize the system. \textbf{d} and \textbf{e}, Low-efficacy insecticides were ineffective in stabilizing the system. The system parameters are listed in Table \eqref{tab:1}.}
    \label{fig:long-term_cabbage}
\end{figure}

\newpage

\begin{figure}[H]
\includegraphics[width=1\linewidth]{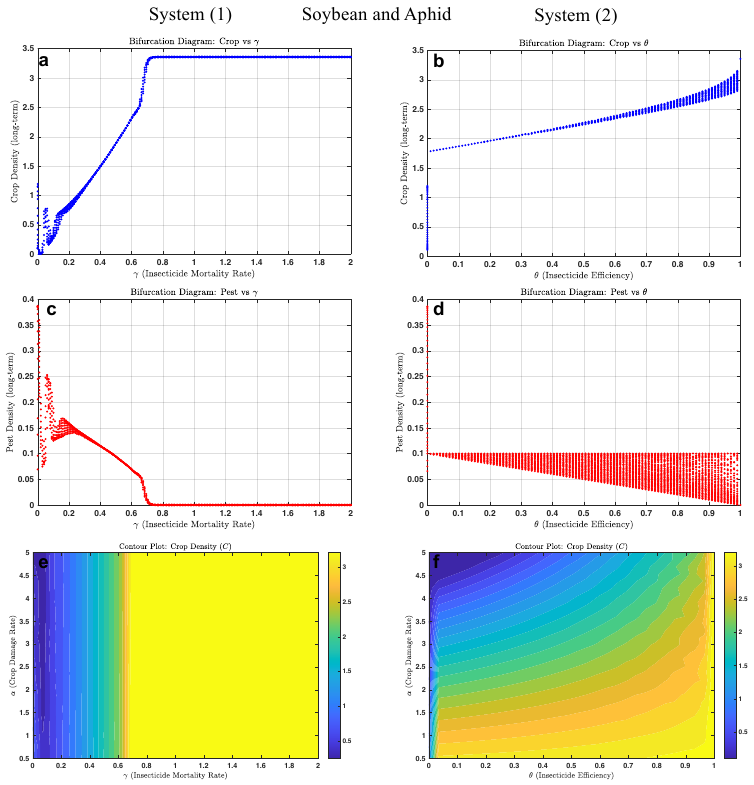}
   \caption{\textbf{Bifurcation analysis of the soybean–aphid system.} The one and two-dimensional bifurcation for the soybean-aphid system. The other system parameters are listed in Table \eqref{tab:1}. (Note: Bifurcation analysis is a mathematical technique used to study how the qualitative behavior of a dynamical system changes as parameters are varied, particularly focusing on transitions in stability \citep{perko2013differential}).}
    \label{fig:soy_biffur}
\end{figure}

\newpage

\end{spacing}

\newpage

\beginsupplement

\begin{center}
	\Large{Supplementary Materials for \\[2ex] \textbf{Fertilizers Fuel, Insecticides Stabilize: Resolving the Paradox of Enrichment in Agriculture}}
\end{center}

\begin{center}
	\textbf{Vaibhava Srivastava$^{1,\dagger,*}$, Jason R. Rohr$^{2,\dagger}$, and Rana D. Parshad$^{3}$}
\end{center}

\begin{center}
	\textit{
		1) Department of Mathematics and Statistics, University of Massachusetts Amherst, Amherst, MA 01003, USA.\\
		2) Department of Biological Sciences, Environmental Change Initiative, Eck Institute of Global Health, University of Notre Dame, Notre Dame, IN 46556, USA.\\
		3) Department of Mathematics, Iowa State University, Ames, IA 50011, USA.
	}
\end{center}

\begin{center}
	\textbf{Corresponding author:} Vaibhava Srivastava (\texttt{vsrivastava@umass.edu})\\
	\textbf{Email:} Jason R. Rohr (\texttt{jrohr2@nd.edu}, \texttt{jasonrohr@gmail.com}), Rana D. Parshad (\texttt{rparshad@iastate.edu})\\[4pt]
	$^\dagger$These authors contributed equally. \quad $^*$Corresponding author.
\end{center}

	\noindent {\bf This PDF file includes:}

	\noindent Sections S1 to S4
	\begin{enumerate}
		\item[\textbf{S.1.}] Background on focal crop-pest systems
		\item[\textbf{S.2.}] Parametrization of case studies
		\item[\textbf{S.3.}] Stability analysis and numerical simulations for Holling type II
		\item[\textbf{S.4.}] Stability analysis and numerical simulations for Holling type III
	\end{enumerate}

	\noindent Tables S1 to S2 \\
	\noindent Figures S1 to S9

\setcounter{page}{37}

\clearpage

\doublespacing

\renewcommand{\thesection}{S.1}
\renewcommand{\thesubsection}{S.1.\arabic{subsection}}
\renewcommand{\thesubsubsection}{S.1.\arabic{subsection}.\arabic{subsubsection}}

\section*{S.1. Background on focal crop-pest systems}

\subsection*{S.1.1. Soybean (\textit{Glycine max (L.) Merr.}) and Soybean Aphid (\textit{Aphis glycines})}

\noindent Soybean (\textit{Glycine max (L.) Merr.}) grain significantly contributes to world agriculture and the economy. As per the report \citep{usdaSoybeans2025}, the global soybean production for $2024/2025$ is projected at $420.78$ million metric tons, reflecting a $6\%$ year-over-year increase. This continues a $10$-year compound growth trend of $3\%$, with an average production of $364.19$ million tons from $2015$ to $2024$. Brazil leads with $40\%$ of global output ($169$ million tons), followed by the United States ($28\%$), Argentina ($12\%$), China ($5\%$), and India ($3\%$). In the United States, in $2024$, soybean production reached a total of $4.37$ billion bushels, a $5$ percent increase from $2023$ \citep{usda2024crop}, and the harvested area, now at $86.1$ million acres, also grew by $5$ percent compared to last year \citep{usda2024crop}. All of these statistics not only underpin the importance of the soybean crop in the agriculture landscape but also make it a vital and significant asset to the global economy. 

Major threats to the soybean come from invasive pests such as soybean aphid (\textit{Aphis glycines}) \citep{ragsdale2007economic, tilmon2011biology}. The soybean aphid causes devastating damage to the crop across the north-central region of the United States and southern Canada \citep{ragsdale2007economic,venette2004assessing}, as well as parts of Asia, including India and China \citep{wu2004soybean,czepak2013first}. In South America—including Brazil\citep{pomari2015helicoverpa,bortolotto2015use,czepak2013first} and Argentina\citep{murua2014first}, and Mid-Southern United States \citep{musser2012soybean,haile2021overview}, the major damage to the soybean crop is caused due to \textit{Helicoverpa armigera}, a agriculture pests that cause yield loss to soybean, particularly in warmer tropical and subtropical regions of the Americas. According to a report\citep{song2006profitability}, soybean aphid infestations in $2003$ affected over $42$ million acres of soybeans in the north-central United States, potentially resulting in yield losses of more than $350$ million bushels and an estimated economic impact of over $\$2.4$ billion if this pest is not managed efficiently and left untreated.

\subsection*{S.1.2. Winter wheat (\textit{Triticum aestivium L.}) and bird cherry-oat aphid (\textit{Rhopalosiphum padi L.})}

\noindent Wheat is one of the most widely cultivated crops in the world, making it important for the global economy. As per the report \citep{usdaWheat2025}, for the $2024/2025$ marketing year, global production in this category is projected at $799.91$ million metric tons, a modest $1\%$ increase from the previous year. The $10$-year average yield is approximately $768.16$ million tons, with a compound annual growth rate of $0.79\%$. China, the European Union, and India are the top producers, accounting for $18\%, 15\%,$ and $14\%$ of global output, respectively. The United States contributes $7\%$ ($53.65$ million tons). In the United States, wheat is the third most cultivated field crop \citep{USDA_ERS_Wheat}, and white wheat, which includes both winter and spring varieties, makes up $12$ to $17$ percent of total wheat production. It is primarily grown in Washington, Oregon, Idaho, Michigan, and New York.

The major threat to winter wheat is from the bird cherry-oat aphid (\textit{Rhopalosiphum padi L.}). It not only damages winter wheat directly, but also acts as an efficient vector of the barley yellow dwarf (BYD) virus (Family: \textit{Luteoviridae}) to winter wheat\citep{riedell1999winter,kieckhefer1992yield}. The aphid's negative impact on the winter wheat during the early stages of the crop leads to a significant losses in the quality of grain and overall yield \citep{riedell1999winter,pike1985development,willocquet2008simulating}. 

\subsection*{S.1.3. Cabbage and diamondback moth (\textit{Plutella xylostella})}

\noindent Cabbage (\textit{Brassica oleracea var. capitata}) originates from the southern and western coastal regions of Europe. It is a widely grown vegetable, cultivated across many countries worldwide. China contributes approximately $47$ to $49\%$ of the total global cabbage production, making it the largest producer worldwide. India is the second-largest producer, with a cabbage output of $9.13$ million metric tons in $2019$, marking an increase of $542$ thousand metric tons since $2015$\citep{TAMUCabbage2021}. Globally, approximately $55$ million tons of fresh cabbage heads are harvested annually from around $2.6$ million hectares of land \citep{FAO2001cabbage}. In the United States, the southeastern parts are an important hub for producing cruciferous vegetables and greenhouse-grown transplants. Among these, cabbage crops significantly contribute to the economies of Florida and Georgia. For example, cabbage production alone generated approximately $\$45.2$ million in Georgia and $\$72.6$ million in Florida in 2022 \citep{USDA_NASS_2023}.

The major pest to the cruciferous vegetables is the diamondback moth (\textit{Plutella xylostella L.}) \citep{talekar1993biology,furlong2013diamondback}. This pest specifically targets crops from the \textit{Brassica} family \citep{talekar1993biology,furlong2013diamondback}. Due to this pest infestation, the estimated annual costs associated have increased significantly. For example, it increased from $\$1$ billion in the 1990s \citep{talekar1993biology} to about $\$4$-$5$ billion in 2010 \citep{zalucki2012estimating}, presenting a challenge to the country's economy and agricultural landscape.

\section*{S.2. Parametrization of case studies}

\subsection*{S.2.1. Soybean (\textit{Glycine max (L.) Merr.}) and Soybean Aphid (\textit{Aphis glycines})}
The growth rate of soybean crops is $0.451 \, \text{day}^{-1}$ \citep{ball2000optimizing}, and the average soybean yield over the 5 years from 2019/20 to 2023/24 is $3.37 \, \text{T/Ha}$ \citep{USDA_Soybean_Production}. Without fertilization, the maximum observed soybean yield was $K_0=2.797$ T/Ha \citep{cordova2020soybean}. Whereas with proper fertilization, the yield varied depending on timing, dosage, etc. In  \cite{cordova2020soybean}, authors claimed that when they used $135$ kg N/Ha at planting and again at the R1/R2 growth stages, they increased yield by approximately $20\%$. However, when they used the reduced application of $45$ kg N/Ha at the later R3/R4 stages, yields were only $8\%$ lower than those achieved with the full $135$ kg N/Ha dosage at R1/R2. So, for two applications of $135$ kg N/Ha $(F=0.559$ T/Ha), whereas for $45$ kg N/Ha, $(F=0.291$ T/Ha)

Soybean aphids' birth rate is $3$ to $8$ per day for $30$ days. The generation time is $7$ to $10$ days. As a result, they can experience exponential growth, with populations potentially doubling every $2$ to $3$ days under favorable \citep{ragsdale2007economic,tilmon2011biology,iowaSoybeanAphid}. Following their exponential growth ($N(t) = N_0 e^{\beta t}$), and their population-doubling effect, we have 
$T_{\text{double}}= \frac{\ln 2 }{\beta}$. Since the population doubles every two days, we have $\beta \approx \frac{\ln 2 }{2} = 0.3466 \hspace{.05in} \text{day}^{-1}$. For the natural mortality rate \citep{pioneerSoybeanAphidMgmt,umnSoybeanAphid}, we have $\delta = 1/(\text{Life span in day})= 1/20 = 0.05 \hspace{.05in}  \text{day}^{-1}.$

To control the pest outbreak, insecticide application is recommended only until the plant reaches the R6 growth stage, and provided that the aphid counts are at least $250$ per plant, at least $80\%$ plants are infested, and the population is increasing \citep{iowaSoybeanAphid,ragsdale2007economic}. The parameterization of the insecticide mortality rate is based on a study conducted at the South Dakota State University Volga Research Farm in 2017 \citep{dierks2019evaluating, menger2020implementation}. The active ingredient, product name, company, labeled rate, and mortality in percentage and rate are given in Table \ref{tab:2} (Supplementary \textcolor{blue}{S.1}).

\subsection*{S.2.2. Winter wheat (\textit{Triticum aestivium L.}) and bird cherry-oat aphid (\textit{Rhopalosiphum padi L.})}

For the wheat crop, the process of fertilization plays an important role. To parameterize the fertilization effect and how the Nitrogen-based fertilizer treatment for winter wheat impacts, we will use the findings of a field study conducted in northern Idaho and eastern Washington \citep{mahler1994nitrogen}.  The major finding of this study was the following: without nitrogen treatment, the yield was $3.35$ T/Ha (i.e, $K_0=3.35$), whereas when compared to when ammonium nitrate was applied either banded below the seed or surface-broadcast in the fall, the yields increased to $7.48$ T/Ha and $7.90$ T/Ha, respectively \citep{mahler1994nitrogen}. This indicates that the yield increased by approximately $123\%$, demonstrating the positive impact of fertilizers and timings on winter wheat production \citep{mahler1994nitrogen}.

For our analysis, we will use the wheat production data from Tennessee. In Tennessee, the average wheat yield is $4.764$ T/Ha, and over the 10 years of 2006-2016, it has ranged from $4.967$ to $3.403$ T/Ha ($50-73$ Bu/Acre) \citep{perkins2018impact}. So, on averaging and using the fact that winter crops are properly fertilized \citep{usda2023fertilizer}, we have total yield $K_0 + F = 4.185$ T/Ha. From the case study \citep{mahler1994nitrogen}, we know that, without the fertilizer, the yield loss would be $123\%$. Hence, we have $K_0=1.877$ T/Ha and $F=2.308$ T/Ha. To determine the growth rate of wheat crops, we will utilize a logistic growth model based on the following information: the maximum yield $K = 4.185$ T/Ha, the initial yield is 0.1 T/Ha, and the time required to reach near-maximum yield is 115 days. At $t = 115$ days, the yield reaches $99\%$ of the maximum, which is $4.1431$ T/Ha. The logistic growth model is expressed as: $Y(t) = \frac{K}{1 + Ae^{-rt}} \implies \quad r=0.0722 \hspace{.1in}\text{day}^{-1}$.

The pest, bird cherry-oat aphid (\textit{Rhopalosiphum padi L.}) grows exponentially and shares many similarities with soybean aphids. Therefore, the growth and mortality rates are assumed to be the same, i.e., $\beta=0.3466  \hspace{.1in} \text{day}^{-1}$ and $ \delta=0.05 \hspace{.1in}\text{day}^{-1}.$

To parameterize the efficacy of the insecticide treatment of aphids, we will use a case study conducted between January 31 and February 25. In this study, the researcher studies the foliar application of insecticide on the pest \citep{perkins2018impact}. The standard treatment involved applying Karate Z at a rate of $27.3$ g active ingredient per hectare (lambda-cyhalothrin, Syngenta Crop Protection). In some tests, Baythroid XL was used instead, at a rate of $13.2$ g active ingredient per hectare (beta-cyfluthrin, Bayer CropScience). This late-winter foliar insecticide application effectively reduced aphid populations by approximately $90\%$ \citep{perkins2018impact}.

\subsection*{S.2.3. Cabbage and Diamondback moth (\textit{Plutella xylostella})}

To parameterize the effect of fertilization on the cabbage, we will refer to an experiment conducted in the Ljubljansko polje, east of Ljubljana, the capital of Slovenia. In this study, researchers examined the impact of various fertilization and irrigation methods on the yield of white cabbage \citep{vsturm2010effect}. The study found that the yield in unfertilized control plots, using the farmers' irrigation practices, was $K_0 = 47$ T/Ha. When the farmers applied their fertilization practices in accordance with the Regulations on Integrated Production of Vegetables and the Technological Instructions for Vegetable Production \citep{MAFF2003, OGRS2002}, in addition to irrigation, the yield increased by $97.87\%$, reaching a total of $93$ T/Ha (i.e., $F=46$ T/Ha). On using a similar logistic type relation to calculate the growth rate of cabbage, we have $r = 0.0911 \hspace{.051in} \text{day}^{-1}$.

The growth rate of diamondback moth is $\beta=0.1477  \hspace{.051in} \text{day}^{-1}$ \citep{kahuthia2008development} and  $\delta=0.1 \hspace{.1in}\text{day}^{-1}$  \citep{al2019role}. To understand the effect and efficacy of multiple registered active ingredients in insecticides on diamondback moth, we used the findings of a field study conducted in Georgia and Florida \citep{dunn2024regional}. The active ingredient, product name, company, labeled rate, and mortality in percentage and rate are given in Table \eqref{tab:3}

We also evaluated the effectiveness of the baculovirus Plutella xylostella Granulovirus (PlxyGV) for managing this pest \citep{gryzwacz2002granulovirus, bhattacharyya2006pest}. In this study, we used research conducted in Kabete and Thika, near Nairobi \citep{bhattacharyya2006pest}. The findings indicated that the application of insecticide led to a significant decrease in the pest population, with reductions of approximately $83\%$ in shaded areas and $79\%$ in open areas \citep{bhattacharyya2006pest}. The active ingredient, product name, company, labeled rate, and mortality in percentage and rate are given in Table \ref{tab:3} (Supplementary \textcolor{blue}{S.1}).

\section*{S.3. Stability analysis and numerical simulations for Holling type II}

\subsection*{S.3.1. Equilibrium analysis}\label{sec:equi}
On setting 

\[ \dfrac{dx}{dt}=\dfrac{dy}{dt}=0, \]
we have

\begin{align*}
rx \left( 1 - \dfrac{x}{K_0 + F} \right) - \dfrac{\alpha x y}{1 + hx}=0 \qquad & \implies x \Big[ r\left( 1 - \dfrac{x}{K_0 + F} \right) - \dfrac{\alpha y}{1 + hx} \Big]=0 \\
\dfrac{\beta x y}{1 + hx} - \delta y - \gamma  y=0 & \implies y \Big[ \dfrac{\beta x }{1 + hx} - \delta  - \gamma \Big]=0.
\end{align*}

So, on setting $x=0$, we have $y=0$. Thus, the crop and pest extinction equilibrium is 

\[ \mathbf{E_1}:= (0,0).\]

On setting $y=0$, we have 
\[\left( 1 - \dfrac{x}{K_0 + F} \right) =0 \implies x= K_0 + F. \]
Thus, the pest extinction equilibrium is 
\[ \mathbf{E_2}:= (K_0+F,0).\] 

And on solving 

\[ r\left( 1 - \dfrac{x}{K_0 + F} \right) - \dfrac{\alpha y}{1 + hx} = 0 \quad \& \quad \dfrac{\beta x }{1 + hx} - \delta  - \gamma = 0, \]
we get a co-existence equilibrium 

\begin{align}\label{co_exist}\tag{Eq S.1.1}
\mathbf{E_3}:= \Big(\frac{\gamma +\delta }{\beta -h (\gamma +\delta )}, \frac{\beta  r \left(-\gamma -\delta +F (\beta -h (\gamma +\delta ))+K_0 (\beta -h (\gamma +\delta ))\right)}{\alpha  \left(F+K_0\right) (\beta -h (\gamma +\delta ))^2} \Big)
\end{align}

Hence, we have three possible equilibria corresponding to the system (1):

\begin{enumerate}
\item \textbf{Extinction:} $\mathbf{E_1}=(x^*,y^*):= (0,0),$
\item \textbf{Pest Extinction:} $\mathbf{E_2}=(x^*,y^*):= (K_0+F,0),$
\item \textbf{Co-existence:} $\mathbf{E_3}=(x^*,y^*):= \Big(\frac{\gamma +\delta }{\beta -h (\gamma +\delta )}, \frac{\beta  r \left(-\gamma -\delta +F (\beta -h (\gamma +\delta ))+K_0 (\beta -h (\gamma +\delta ))\right)}{\alpha  \left(F+K_0\right) (\beta -h (\gamma +\delta ))^2} \Big)$
\end{enumerate}
\noindent \textbf{To make sure the equilibria are ecologically relevant, we will always assume that $\beta > h (\gamma +\delta )$.}

\subsection*{S.3.2. Boundedness and positivity}\label{sec:bdd_pst}

\begin{theorem}\label{thm:bdd_pst}
The solution $(x(t),y(t))$ of system (1), with positive initial values $(x(0)>0,y(0)>0)$, are positive and bounded for all $t\ge0.$
\end{theorem}
\begin{proof}
\textbf{Positivity:} We can rewrite the system (1), 
\begin{align*}
	\dfrac{dx}{dt} &= x  \Big[ r\left( 1 - \dfrac{x}{K_0 + F} \right) - \dfrac{\alpha y}{1 + hx} \Big] \equiv x f(x,y), \\[12pt]
	\dfrac{dy}{dt} &=  y \Big[ \dfrac{\beta x}{1 + hx} - \delta - \gamma \Big] \equiv x g(x,y)
\end{align*}

Since, the initial conditions are positive $x(0)>0$ and $y(0)>0$, the solution of system (1) are defined 

\[
\displaystyle
x(t) = x(0) e^{ \int_0^t f(x(s), y(s))\, ds} > 0,
\]
and 
\begin{align*}
	y(t) = y(0) e^{\int_0^t g(x(s),y(s)) ds} >0.
\end{align*}
Hence, the solutions are positive.

\noindent\textbf{Boundedness:}Let's define the function, $w(t)=\beta x + \alpha y$. On calculating, 
\begin{align*}
	\dfrac{dw}{dt} &= \beta \dfrac{dx}{dt} + \alpha \dfrac{dy}{dt} \\
	&= \beta  \Big[ rx \left( 1 - \dfrac{x}{K_0 + F} \right) - \dfrac{\alpha x y}{1 + hx} \Big] + \alpha \Big[ \dfrac{\beta x y}{1 + hx} - \delta y - \gamma  y \Big] \\
\end{align*}
Since, the prey follows logistic growth, we have $x(t) \le K_0 + F + \epsilon$ for any $\epsilon>0$ and  sufficiently large t\citep{murray2007mathematical}. Hence
\begin{align*}
	\dfrac{dw}{dt} &\le \beta  (K_0 + F +\epsilon) - \alpha (\delta + \gamma) y \\
	& \le 2 \beta (K_0 + F +\epsilon) - \min  \{ \delta + \gamma, r \} w.
\end{align*}
So, \[ w \le \dfrac{2 \beta (K_0 + F +\epsilon)}{\min \{ \delta + \gamma, r \}},\]
for $t$ large enough, hence $w$ is bounded. Since, $w(t)=\beta x + \alpha y$, so the solution $(x(t),y(t))$ of system (1) is bounded.
\end{proof}
\subsection*{S.3.3. Stability analysis}\label{sec:stab}

\begin{theorem}\label{thm_extinction}
The equilibrium $\mathbf{E_1}:=(0,0)$ is always unstable.
\end{theorem}

\begin{proof}
\begin{align*}
	J(\mathbf{E_1}):=\left(
	\begin{array}{cc}
		r & 0 \\
		0 & -\gamma -\delta  \\
	\end{array}
	\right)
\end{align*}

The eigenvalues are 
\begin{align*}
	\lambda_1 &= r\\
	\lambda_2 &= -\gamma I -\delta \\
\end{align*}
It is unstable and saddle since the eigenvalue $\lambda_{1}$ is always positive, whereas $\lambda_2$ is always negative (since the parameters $I,\delta\ge 0)$\citep{perko2013differential}.
\end{proof}

\begin{theorem}\label{thm_pest_extinction}
The pest extinction equilibrium $\mathbf{E_2}=(K_0+F,0)$ is stable if 
\[ \beta (K_0+F) < \delta + \gamma.\]
\end{theorem}

\begin{proof}

\begin{align*}
	J(\mathbf{E_2}):&=\left(
	\begin{array}{cc}
		r-\dfrac{2 r x^*}{F+K_0} & -\dfrac{\alpha  x^*}{h x^*+1} \\[14pt]
		0 & -\gamma -\delta +\dfrac{\beta  x^*}{h x^*+1} \\
	\end{array}
	\right)
\end{align*}

Since, $x^*=F+K_0$, we have

\begin{align*}
	J(\mathbf{E_2}):&=\left(
	\begin{array}{cc}
		-r & -\dfrac{\alpha  \left(F+K_0\right)}{h \left(F+K_0\right)+1} \\[14pt]
		0 & -\gamma -\delta +\dfrac{\beta  \left(F+K_0\right)}{h \left(F+K_0\right)+1} \\
	\end{array}
	\right)
\end{align*}

The eigenvalues are 
\begin{align*}
	\lambda_1 &= -r\\
	\lambda_2 &= -\gamma -\delta +\dfrac{\beta  \left(F+K_0\right)}{h \left(F+K_0\right)+1}  \\
\end{align*}

As per the assumption, we have $\beta (K_0+F) < \delta + \gamma$, and using the density of real number, we have

\begin{align*}
	\dfrac{\beta  \left(F+K_0\right)}{h \left(F+K_0\right)+1} < \beta  \left(F+K_0\right) <\gamma +\delta \implies \lambda_2<0.
\end{align*}

Since $r$, being the growth rate, is always positive, $\lambda_1<0$. And, from our assumption, $\lambda_2<0$. Hence, the equilibrium $\mathbf{E_2}$ is locally stable\citep{perko2013differential}.
\end{proof}

Similarly, we can find the parametric restriction for which the $\mathbf{E_2}$ is unstable.
\begin{theorem}
The pest extinction equilibrium $\mathbf{E_2}$ is unstable if 
\[ \frac{\beta  (K_0+F)}{h (K_0+F)+1} > \delta + \gamma.\]
\end{theorem}

\begin{theorem}\label{thm_global_extinction}
The system $\mathbf{E_2}=(K_0+F,0)$ is globally asymptotically stable if 
\[\beta + \beta (K_0 +F)<\delta + \gamma.\]
\end{theorem}

\begin{proof}
Let $\Omega:=\{ (x,y) \in \mathbb{R}^2 : x>0,y>0\}=\mathbb{R}^2_{+}$, and define
\[ V: \mathbb{R}^2_{+} \to \mathbb{R} \quad \text{such that} \quad V(x,y)=\dfrac{1}{2} (x-(F+K_0))^2) + \dfrac{1}{2} y^2 + y.\]

We want to use the Lyapunov function $V(x,y)$ to show global stability for $\mathbf{E_2}$. On calculating,
\begin{align*}
	\dfrac{dV}{dt} & = ((x-(F+K_0)) \dfrac{dx}{dt} + y \dfrac{dy}{dt} + \dfrac{dy}{dt} \\
	& = (x-(F+K_0)) \Big[ rx \left( 1 - \dfrac{x}{K_0 + F} \right) - \dfrac{\alpha x y}{1 + hx} \Big] + y \Big[ \dfrac{\beta x y}{1 + hx} - \delta y - \gamma  y \Big] + \Big[ \dfrac{\beta x y}{1 + hx} - \delta y - \gamma  y \Big] \\
	& = -rx (x-(F+K_0))^2 - \dfrac{\alpha x (x-(F+K_0)) y}{1+hx} + \dfrac{\beta x y^2}{1+hx} -(\delta+\gamma) y^2 + \dfrac{\beta x y}{1+h x} - (\delta + \gamma) y \\
	& \le  \alpha (F+K_0) y + \beta y^2 -(\delta+\gamma) y^2 + \beta  y - (\delta + \gamma) y \\
	&= y \Big( \alpha (F+K_0) + \beta - (\delta + \gamma) \Big) + y^2 \Big(\beta - (\delta + \gamma) \Big).
\end{align*}

From the assumption, we have 
\[ \beta + \beta (K_0 +F)<\delta + \gamma \quad \& \beta <\delta + \gamma.\]

Hence, the $\dfrac{dV}{dt}<0$ in $\Omega$, which gives $\mathbf{E_2}$ is globally stable\citep{perko2013differential}.
\end{proof}

\begin{theorem}\label{thm_co-existence}
The co-existence equilibrium $\mathbf{E_3}$ is locally stable if 
\[r<\dfrac{2 r x^*}{F+K_0}+\dfrac{\alpha  y^*}{(h x^* +1)^2}.\]

\end{theorem}

\begin{proof}
\begin{align*}
	J(\mathbf{E_3}):=\left(
	\begin{array}{cc}
		-\dfrac{2 r x^*}{F+K_0}-\dfrac{\alpha  y^*}{(h x^* +1)^2}+r & -\dfrac{\alpha  x^* }{h x^* +1} \\[14pt]
		\dfrac{\beta  y^*}{(h x^* +1)^2} & -\gamma -\delta +\dfrac{\beta  x^* }{h x^* +1} \\
	\end{array}
	\right)
\end{align*}

Since, $x^*$ and $y^*$ solve the following equations, 

\[ r\left( 1 - \dfrac{x^*}{K_0 + F} \right) - \dfrac{\alpha y^*}{1 + hx^*} = 0 \quad \& \quad \dfrac{\beta x^* }{1 + hx^*} - \delta  - \gamma = 0, \]

we have 

\begin{align*}
	J(\mathbf{E_3}):=\left(
	\begin{array}{cc}
		-\dfrac{2 r x^*}{F+K_0}-\dfrac{\alpha  y^*}{(h x^* +1)^2}+r & -\dfrac{\alpha  (\gamma +\delta )}{\beta } \\[14pt]
		\dfrac{\beta  y^*}{(h x^* +1)^2} & 0 \\
	\end{array}
	\right)
\end{align*}

The co-existence equilibrium $\mathbf{E_3}$ is stable if $\det (J(\mathbf{E_3})>0$ and $\mathrm{Tr} (J(\mathbf{E_3}))<0$, where
\begin{align*}
	\mathrm{Tr} (J(\mathbf{E_3})) &= -\dfrac{2 r x^*}{F+K_0}-\dfrac{\alpha  y^*}{(h x^* +1)^2}+r  \\
	\det (J(\mathbf{E_3}) &= \dfrac{\alpha  (\gamma +\delta )}{\beta } \dfrac{\beta  y^*}{(h x^* +1)^2}
\end{align*}

From the positivity of parameters and using the fact that $x^*,y^*>0$, we have $\det (J(\mathbf{E_3})>0.$ Under the assumption 

\[ r<\dfrac{2 r x^*}{F+K_0}+\dfrac{\alpha  y^*}{(h x^* +1)^2},\]

we have $\mathrm{Tr} (J(\mathbf{E_3}))<0$, hence $\lambda_{1,2}$ are negative. Hence, the co-existence equilibrium $\mathbf{E_3}$ is locally stable.
\end{proof}

\section*{S.4. Stability analysis and numerical simulations for Holling type III}\label{sec:holling_type3}
Let $(x)$ and $(y)$ denote the crop and pest in the agroecosystems. In our system, a crop $(x)$ follows a logistic growth with a growth rate $(r)$, and it can reach to carrying capacity $(K_0)$ without fertilizer. Using fertilizer results in an increased yield of $(F)$. The pest $(y)$ attacks the crop with rate $\alpha$, with a growth rate of $(\beta)$, and with a natural mortality rate of $(\delta)$. The farmer's use of insecticide results in insecticide-induced mortality ($\gamma$). The $\gamma$ is calculated as $-\text{ln(survival})/\text{time}$.  The details of the parameters and their units are in Table 1. The crop-pest interaction using the system of nonlinear ordinary differential equations (ODEs):

\begin{equation}\tag{S.1}
\begin{array}{l}
	\dfrac{dx}{dt} = rx \left( 1 - \dfrac{x}{K} \right) - \dfrac{\alpha x^k y}{1 + hx^k}, \\[12pt]
	\dfrac{dy}{dt} = \dfrac{\beta x^k y}{1 + hx^k} - \delta y - \gamma y.
\end{array}.
\label{eq:study_holling_type3}
\end{equation}

Additionally, we have defined a pulse version of the model \eqref{eq:study_holling_type3} that relies on a pest threshold, meaning that insecticide application is only necessary when the pest population exceeds this threshold. 
\begin{equation}\tag{S.2}
\begin{array}{ll}
	\dfrac{dx}{dt} = rx \left( 1 - \dfrac{x}{K_0 + m F} \right) - \dfrac{\alpha x^k y}{1 + hx^k}, \\[12pt]
	\dfrac{dy}{dt} = \dfrac{\beta x^k y}{1 + hx^k} - \delta y , \\[12pt]
	\text{If } y > y_{\text{thresh}} \text{ (threshold)}, \hspace{.05in} \text{then } y \rightarrow (1-\theta) y,
\end{array}.
\label{eq:study_holling_type3_control}
\end{equation}
where $\theta$ is the mortality (in percentage) caused by insecticide to the pest. The system \eqref{eq:study_holling_type3_control} helps us set a pest threshold, which triggers the insecticide spraying usage. 

The stability analysis is similar to what we did for Holling type II response\citep{murray2007mathematical,holling1959components,holling1959some,perko2013differential}.

\newpage

\begin{table}[H]
\centering
\caption{\textbf{Treatment composition and associated mortality rates at the SDSU Volga Research Farm in 2017.}}
\begin{adjustbox}{max width=\textwidth}
	\begin{tabular}{|c|p{4cm}|p{4cm}|p{2cm}|p{2cm}|p{2cm}|}
		\hline
		\textbf{Treatment} & \textbf{Product (Active Ingredients)} & \textbf{Company (Location)} & \textbf{Rate (mL/Ha)} & \textbf{Mortality Efficacy ($\theta$) (\%)} & \textbf{Mortality Rate ($\gamma=-\ln(1-\theta)$}\\
		\hline
		1 & Besiege\textsuperscript{\textregistered} (4.63\% lambda-cyhalothrin, 9.26\% chlorantraniliprole) & Syngenta Crop Protection, Inc., Greensboro, NC & 365 & 85 & 1.897\\
		\hline
		2 & Warrior II\textsuperscript{\textregistered} (22.8\% lambda-cyhalothrin) & Syngenta Crop Protection, Inc., Greensboro, NC & 70 & 53 & 0.755\\
		\hline
		3 & Lannate\textsuperscript{\textregistered} (29\% methomyl) & DuPont, Johnston, IA & 615 & 68 & 1.139 \\
		\hline
		4 & Leverage 360\textsuperscript{\textregistered} (21\% imidacloprid, 10.5\% beta-cyfluthrin) & Bayer Crop Science, Pittsburgh, PA & 204 & 33 & 0.4 \\
		\hline
	\end{tabular}
\end{adjustbox}
\label{tab:2}
\caption*{\footnotesize \\
	Summary of treatment formulations and observed mortality rates from field trials conducted at the South Dakota State University Volga Research Farm during the 2017 soybean growing season\citep{dierks2019evaluating,menger2020implementation}}
\end{table}

\newpage

\begin{table}[H]
\centering
\caption{\textbf{Insecticide products evaluated for efficacy against diamondback moth (Plutella xylostella) in Georgia and Florida.}}
\begin{adjustbox}{max width=\textwidth}
	\begin{tabular}{|p{4.5cm}|p{5.5cm}|p{2cm}|p{1cm}|p{1.5cm}|p{1.2cm}|}
		\hline
		\textbf{Active Ingredient} & \textbf{Product Name (Company)} & \textbf{Product ha$^{-1}$} & \textbf{ai ha$^{-1}$ (g)} & \textbf{Efficacy (\%)} & \textbf{Rate} \\ \hline
		Naled & Dibrom 8E (AMVAC Chemical Corp.) & 2340 ml & 2240 & 87.10 & 2.048\\ \hline
		\textit{B. thuringiensis} var. \textit{aizawai} & XenTari DF (Valent USA) & 1680 g & 907 & 68.40 & 1.152 \\ \hline
		Bifenthrin & Brigade 2EC (FMC Corp.) & 467 ml & 112 & 68.00 & 1.139\\ \hline
		\textit{B. thuringiensis} var. \textit{kurstaki} & DiPel DF (Valent USA) & 1120 g & 605 & 3.70 & 0.038 \\ \hline
		Tolfenpyrad & Torac 1.29EC (Nichino America) & 1530 ml & 237 & 70.10 & 1.204 \\ \hline
		Methomyl & Lannate 2.4LV (Corteva AgriScience) & 3510 ml & 1010 & 56.30 & 0.828 \\ \hline
		Spinetoram & Radiant 1SC (Corteva AgriScience) & 730 ml & 87.5 & 85.60 & 1.938 \\ \hline
		Indoxacarb & Avaunt 30WDG (FMC Corp.) & 256 g & 76.7 & 42.00 & 0.545 \\ \hline
		Chlorantraniliprole & Coragen 1.67SC (FMC Corp.) & 365 ml & 73.0 & 40.50 & 0.519 \\ \hline
		Cyclaniliprole & Harvanta 50SL (Summit Agro USA) & 1200 ml & 59.9 & 17.40 & 0.191 \\ \hline
		Cyantraniliprole & Exirel 0.83SC (FMC Corp.) & 986 ml & 98.1 & 37.90 & 0.476 \\ \hline
		Emamectin benzoate & Proclaim 5WDG (Syngenta Crop Protection) & 351 g & 15.5 & 88.70 & 2.18\\ \hline
	\end{tabular}
\end{adjustbox}
\caption*{\footnotesize \\
	List of insecticide products used in the maximum dose efficacy survey against field populations of Plutella xylostella from Georgia and Florida. Products are categorized by their active ingredients, company, rate, and efficacy determined through standardized bioassays to inform regional resistance patterns and pest management strategies\citep{dunn2024regional}.}
\label{tab:3}
\end{table}

\newpage

\begin{figure}[H]
\includegraphics[width=1\linewidth]{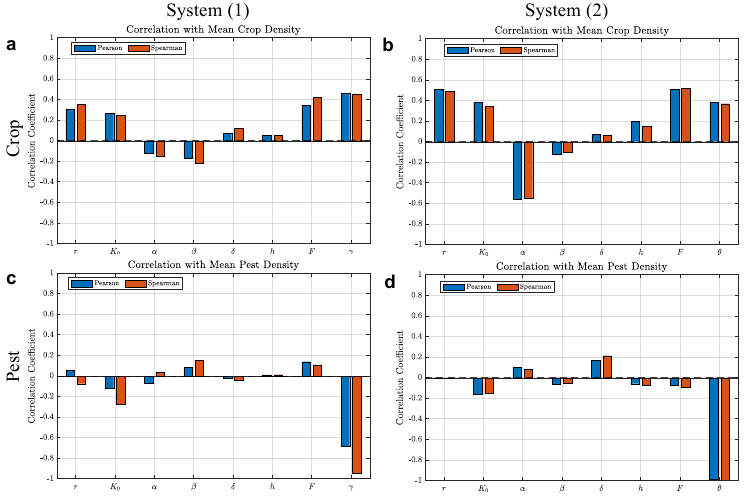}
\caption{\textbf{Sensitivity analysis of the soybean–aphid crop–pest system using Partial Rank Correlation Coefficients (PRCCs).}. The PRCC analysis examines the influence of key model parameters on two primary outputs for systems (1) and (2): \textbf{a} and \textbf{b}, Soybean yield. \textbf{c} and \textbf{d}, Aphid population. Parameters with higher absolute PRCC values exert greater influence on system behavior. Positive values indicate parameters that increase the output when increased, while negative values reflect inhibitory effects.  The system parameters are listed in Table (1).}
\label{fig:prcc}
\end{figure}

\newpage

\begin{figure}[H]
\includegraphics[width=1\linewidth]{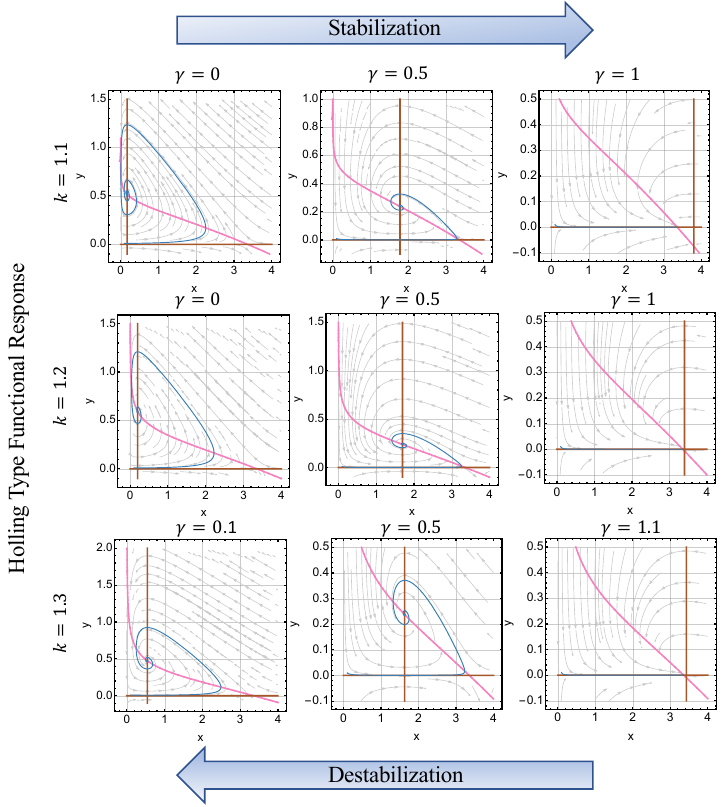}
\caption{\textbf{Phase diagram illustrating the impact of insecticide-induced mortality parameters $(\gamma)$ on the stabilization and destabilization of soybean–aphid systems for Holling type III functional response.} Phase diagram showing the insecticide-induced mortality parameter $(\gamma)$, which stabilizes the soybean-aphid system dynamics for Holling type III functional response: (1) $k=1.1$, (2) $k=1.2$, and (3) $k=1.3$. In the phase diagram, the magenta curve represents and the brown curve represents the nullcline, the intersection of nullclines represents equilibrium points, and the blue curve illustrates the population trajectory. The system parameters are listed in Table 1 and \eqref{tab:2}. (Note: A phase plot is a geometric visualization of the trajectories of an ordinary differential equation (ODE), showing how the system evolves over time from a given initial condition. A nullcline is a curve in the phase plane where the rate of change of one of the variables in the ODE is zero.) }
\label{fig:type3_phase}
\end{figure}

\newpage
\begin{figure}[H]
\includegraphics[width=1\linewidth]{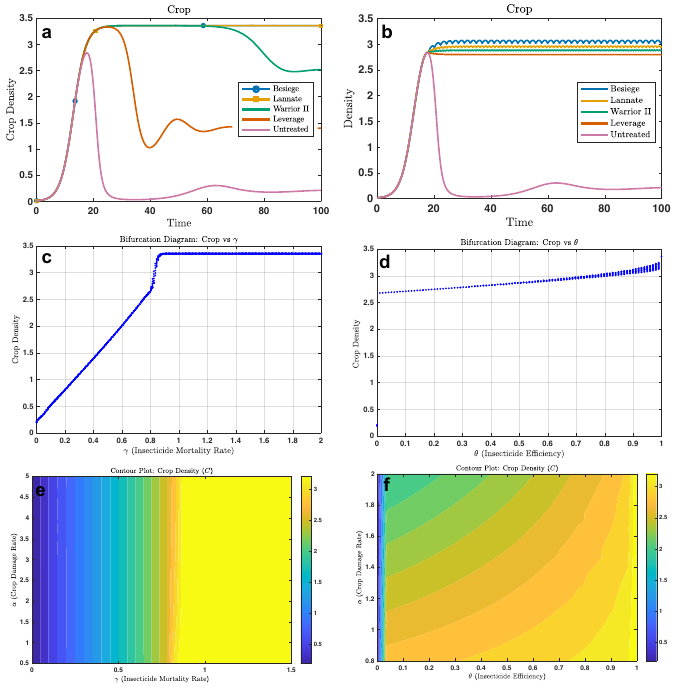}
\caption{\textbf{The long-term and bifurcation analysis of soybean-aphid system with Holling type III response with $k=1.2$.} \textbf{a} and \textbf{b}, Long-term dynamical behavior of the soybean-aphid crop-pest model, showcasing how the application of insecticides stabilizes soybean yield for both continuous and pulse control systems. \textbf{c}-\textbf{f}, One-dimensional and two-dimensional bifurcation analyses of the soybean-aphid system. The system parameters are listed in Table 1 and \eqref{tab:2}. Note: Bifurcation analysis is a mathematical technique used to study how the qualitative behavior of a dynamical system changes as parameters are varied, particularly focusing on transitions in stability \citep{perko2013differential})}
\label{fig:type3_long}
\end{figure}

\newpage

\begin{figure}[H]
\includegraphics[width=1\linewidth]{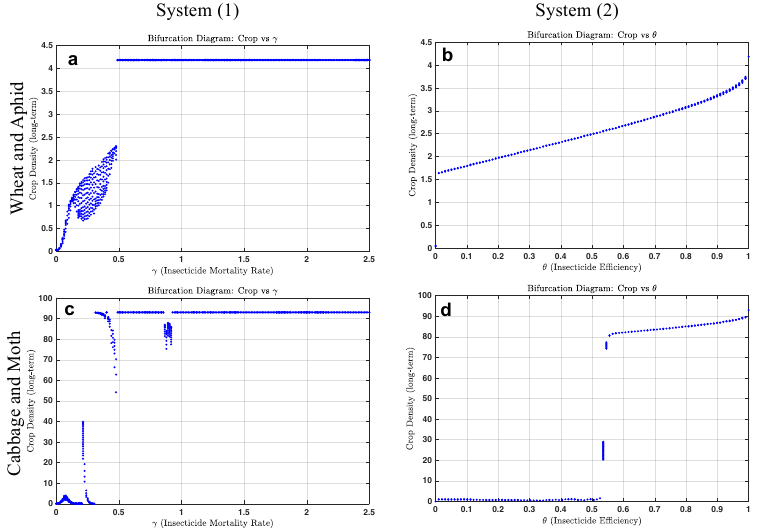}
\caption{\textbf{One-dimensional bifurcation analysis of the wheat-aphid and cabbage-moth system.} The one-dimensional bifurcation for crop yield for the wheat-aphid and cabbage-moth system. The other system parameters are listed in Table 1. (Note: Bifurcation analysis is a mathematical technique used to study how the qualitative behavior of a dynamical system changes as parameters are varied, particularly focusing on transitions in stability \citep{perko2013differential}).}
\label{fig:1d_crop}
\end{figure}

\newpage

\begin{figure}[H]
\includegraphics[width=1\linewidth]{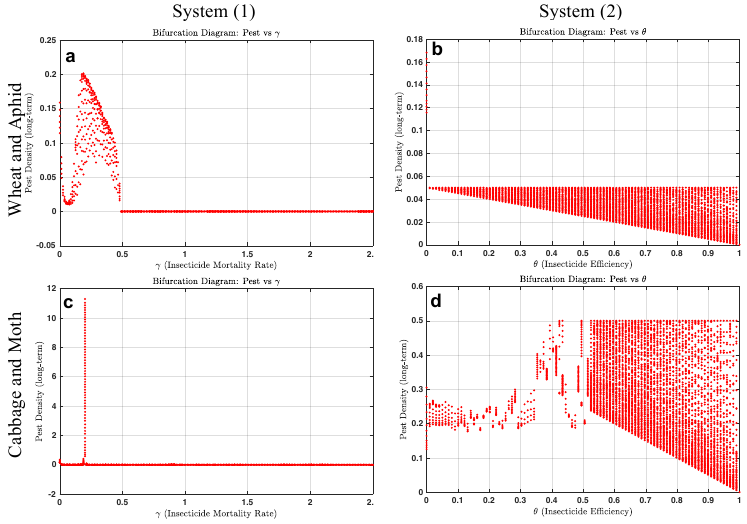}
\caption{\textbf{One-dimensional bifurcation analysis of the wheat-aphid and cabbage-moth system.} The one-dimensional bifurcation for pest population for the wheat-aphid and cabbage-moth system. The other system parameters are listed in Table 1. (Note: Bifurcation analysis is a mathematical technique used to study how the qualitative behavior of a dynamical system changes as parameters are varied, particularly focusing on transitions in stability \citep{perko2013differential}).}
\label{fig:1d_pest}
\end{figure}

\newpage

\begin{figure}[H]
\includegraphics[width=1\linewidth]{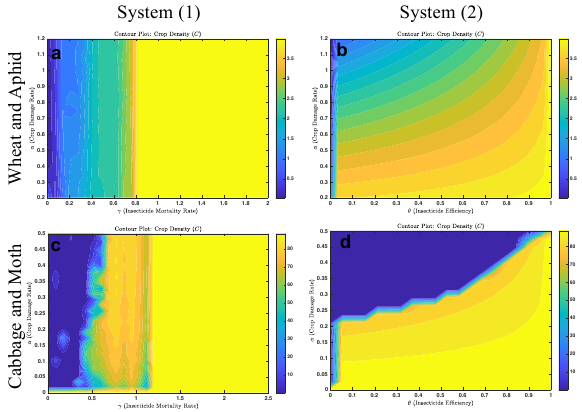}
\caption{\textbf{Two-dimensional bifurcation analysis of the wheat-aphid and cabbage-moth system.} The two-dimensional bifurcation for the wheat-aphid and cabbage-moth system. The other system parameters are listed in Table 1. (Note: Bifurcation analysis is a mathematical technique used to study how the qualitative behavior of a dynamical system changes as parameters are varied, particularly focusing on transitions in stability \citep{perko2013differential}).}
\label{fig:2d}
\end{figure}

\newpage

\begin{figure}[H]
\includegraphics[width=1\linewidth]{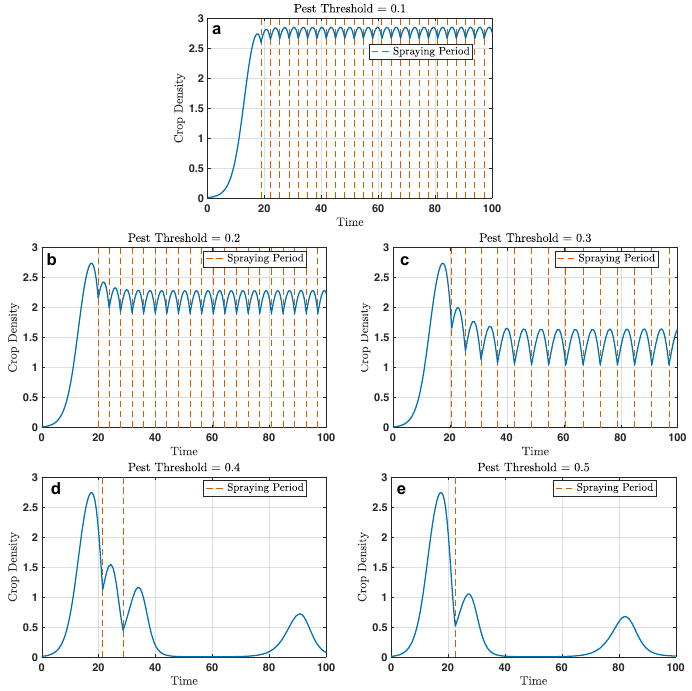}
\caption{\textbf{Soybean–aphid dynamics under threshold-based insecticide application (90\% mortality on spray day).} Simulation of soybean-aphid system dynamics under insecticide efficacy scenarios, assuming 90\% aphid mortality only on the day of spraying, with insecticide application triggered when the aphid population exceeds the threshold $y_{\text{thres}}$. The other system parameters are listed in Table 1.}
\label{fig:spray_one}
\end{figure}

\newpage

\begin{figure}[H]
\includegraphics[width=1\linewidth]{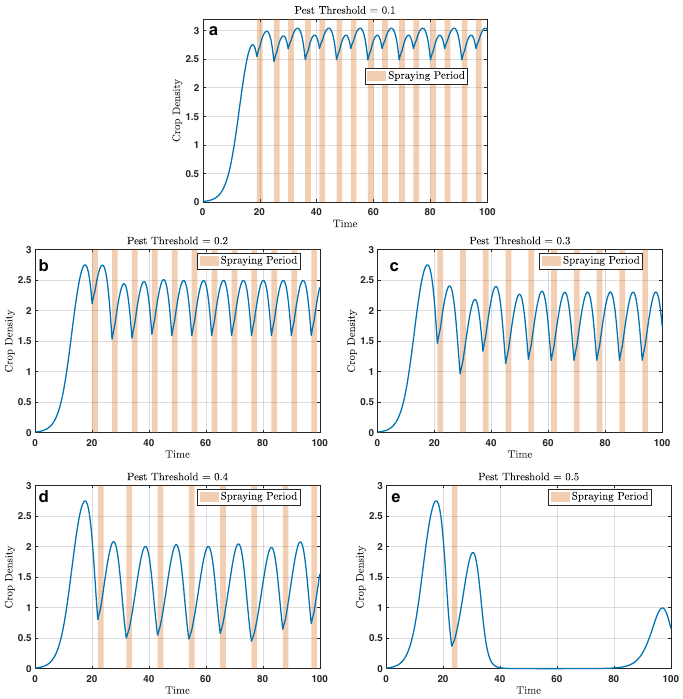}
\caption{\textbf{Soybean–aphid dynamics under threshold-based insecticide application (90\%  aphid mortality on the day of spraying and 70\% the following day).} Simulation of soybean and aphid population dynamics under insecticide efficacy scenarios, assuming 90\%   aphid mortality on the day of spraying and 70\% the following day, with insecticide application triggered when the aphid population exceeds the threshold $y_{\text{thres}}$. The other system parameters are listed in Table 1.}
\label{fig:spray_two}
\end{figure}

\newpage

\begin{figure}[H]
\includegraphics[width=1\linewidth]{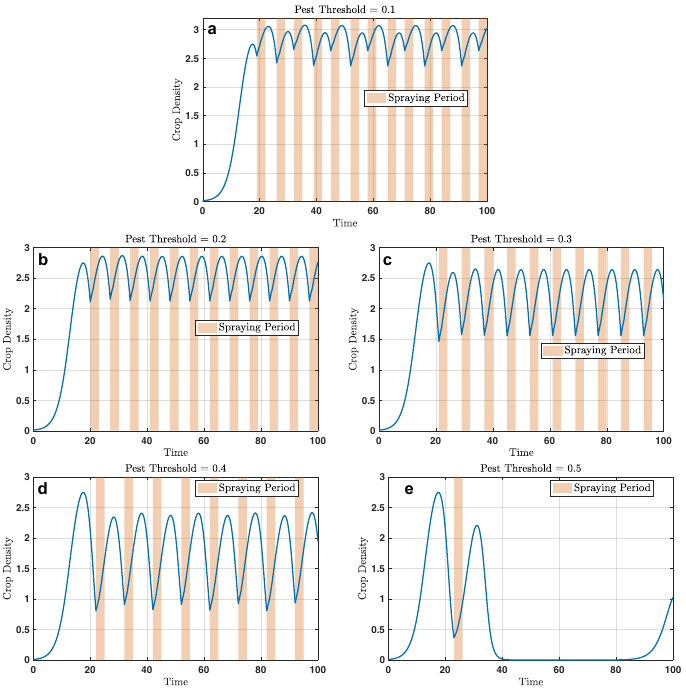}
\caption{\textbf{Soybean–aphid dynamics under threshold-based insecticide application (aphid mortality rates of 90\%, 80\%, and 50\% on three consecutive days).} Simulation of soybean and aphid population dynamics under insecticide efficacy scenarios, assuming aphid mortality rates of 90\%, 80\%, and 50\% on three consecutive days, with insecticide application triggered when the aphid population exceeds the threshold $y_{\text{thres}}$. The other system parameters are listed in Table 1.}
\label{fig:spray_three}
\end{figure}

\newpage

\end{document}